\documentclass[10pt,reqno]{amsart}
\usepackage[UKenglish]{babel}

\usepackage{tikz-cd}
\usepackage[utf8]{inputenc}
\usepackage{amsmath}
\usepackage{amssymb}
\usepackage{amsthm}
\usepackage{mathtools}
\usepackage{url}
\usepackage{mdwlist}

\usepackage{doi}

    \usepackage[numbers,sort&compress]{natbib}

\usepackage{xparse,aliascnt,hyperref,bookmark}

\NewDocumentCommand{\xnewtheorem}{m o m}
 {%
  \IfNoValueTF{#2}
   {\newtheorem{#1}{#3}}
   {%
    \newaliascnt{#1}{#2}%
    \newtheorem{#1}[#1]{#3}%
    \aliascntresetthe{#1}%
    \expandafter\newcommand\csname #1autorefname\endcsname{#3}%
   }%
 }
\addto\extrasUKenglish{
    
}

\hypersetup{
  colorlinks=false,
   linkcolor=blue,
   filecolor=magenta,      
   urlcolor=cyan,
}

\numberwithin{equation}{section}

\def \Lp {\mathrm L}  
\def\Ln {\mathcal L}  
\def \R {\mathbb R}

\newcommand*{\blank}{{\makebox[1ex]{\textbf{$\cdot$}}}}

\xnewtheorem{theorem}{Theorem}[section]
\xnewtheorem{lemma}[theorem]{Lemma}
\xnewtheorem{definition}[theorem]{Definition}
\xnewtheorem{corollary}[theorem]{Corollary}
\xnewtheorem{proposition}[theorem]{Proposition}
\xnewtheorem{fact}[theorem]{Fact}
\theoremstyle{remark}
\newtheorem*{remark}{Remark}

\title[On the stability of the area law for the e.e. of the Landau Hamiltonian]{On the stability of the area law for the entanglement entropy of the Landau Hamiltonian}
\author{Paul Pfeiffer}
\address{Fakultät für Mathematik und Informatik, FernUniversität in Hagen, Universitätsstraße 1, 58097, Hagen, Germany}
\email{paul.pfeiffer@fernuni-hagen.de}
\date{10th May 2020}
\allowdisplaybreaks

\begin{document}

\begin{abstract}

We consider the two-dimensional ideal Fermi gas subject to a magnetic field which is perpendicular to the Euclidean plane $\mathbb R^2$ and whose strength $B(x)$ at $x\in\mathbb R^2$ converges to some $B_0>0$ as $\|x\|\to\infty$. Furthermore, we allow for an electric potential $V_\varepsilon$ which vanishes at infinity. They define the single-particle Landau Hamiltonian of our Fermi gas (up to gauge fixing). Starting from the ground state of this Fermi gas with chemical potential $\mu\ge B_0$ we study the asymptotic growth of its bipartite entanglement entropy associated to $L\Lambda$ as $L\to\infty$ for some fixed bounded region $\Lambda\subset\mathbb R^2$. We show that its leading order in $L$ does not depend on the perturbations $B_\varepsilon := B_0 - B$ and $V_\varepsilon$ if they satisfy some mild decay assumptions. Our result holds for all $\alpha$-R\' enyi entropies $\alpha>1/3$; for $\alpha\le 1/3$, we 
have to assume in addition some differentiability of the perturbations $B_\varepsilon$ and $V_\varepsilon$. The case of a constant magnetic field $B_\varepsilon = 0$ and with $V_\varepsilon= 0$ was treated recently for general $\mu$ by Leschke, Sobolev and Spitzer. Our result thus proves the stability of that area law under the same regularity assumptions on the boundary $\partial \Lambda$.
\end{abstract}
\maketitle

\tableofcontents

\section{Introduction}

Bipartite entanglement entropy is an important quantity that measures correlations of particles inside a given region with the particles outside that region. These non-trivial correlations are solely due to the Fermi--Dirac statistics of the particles involved. In recent years there has been considerable interest and progress in quantifying these correlations. Mathematicians and physicists alike realized fascinating connections between the large scale asymptotics of entanglement entropy and certain semi-classical asymptotic formulas of traces of certain operators, mostly Toeplitz operators in the discrete case and Wiener--Hopf operators in the continuous case. 

In the discrete setting, Jin and Korepin related the Fisher--Hartwig conjecture of Toeplitz matrices to the scaling of the entanglement entropy in the $XY$-chain in a transverse magnetic field in \cite{Jin_2004}. More relevant to our continuous setting here is the discovery of Gioev and Klich \cite{Gioev_2006} that a conjecture by Harold Widom (proved by Alexander V. Sobolev \cite{sobolev2013pseudo}) gives the precise leading asymptotic growth of the bipartite entanglement entropy in ground states of the free Fermi gas. It displays a logarithmically enhanced area law of the order $L^{d-1}\ln(L)$, where $L$ is a scaling parameter, see below. In \cite{Leschke_2014}, this was finally proved by Leschke, Sobolev and Spitzer. In  \cite{muller2021stability}, \cite{muller2020stability},  M\"uller and Schulte proved that this law is stable under a perturbation by a compactly supported potential. The line of proof in their first paper  is also important for our model here. 

A ground state of a non-interacting fermions on $\R^2$ with single-particle Hamiltonian $H$ as in our model is given by the (Fermi) spectral projection $1_{\le\mu}(H)$, where $\mu\in\mathbb R$. The function $1_{\le \mu}$ is the indicator function of the set $(-\infty,\mu]\subset\R$ and the number $\mu$ is called the Fermi energy. Let $\alpha>0$ and let $h_\alpha$ be the R\' enyi  entropy function, see \eqref{h_a}. For a given bounded region $\Lambda\subset \R^2$ we denote by $1_\Lambda$ the (multiplication operator associated to the) indicator function on $\Lambda$. Then we define the local entropy (or entanglement entropy) $S_\alpha(\Lambda)$ to be the (usual Hilbert space) trace of $h_\alpha$ applied to the spatially to $\Lambda$ reduced Fermi projection, that is,
\begin{equation}
 S_\alpha(\Lambda) :=\operatorname{tr} h_\alpha(1_\Lambda 1_{\le \mu}(H) 1_\Lambda)\,. \label{intro ee2}
\end{equation}
At positive temperature a definition of entanglement entropy or mutual information needs to be amended, see \cite{Leschke_2016}.

For a fixed region $\Lambda$, it is generally hard or impossible to calculate the entropy. However, if we introduce a scaling parameter $L > 0$ and consider the leading order asymptotic expansion of \eqref{intro ee2}  with $\Lambda$ replaced by $L \Lambda$ for $L \to \infty$, there are interesting results. They all assume some kind of regularity of the boundary $\partial \Lambda$, assume the Hamiltonian $H$ to be of a certain form, and may restrict to the case $\alpha=1$. For $H = -\nabla^2 + V$, with some assumptions on $V$, there are results presented in \cite{Elgart2017,Leschke_2014,Leschke_2016,muller2020stability,Muller2020,zbMATH06865058,Pfirsch2018}.

In this paper, we consider the Hamiltonian $H = (-i\nabla -A)^2 + V_\varepsilon$, which is a slight perturbation of the Landau Hamiltonian $H_0$ for a constant magnetic field and no electric field, see \eqref{def H0} and \eqref{def H}. Entanglement entropy of the ground state of the latter Landau Hamiltonian (for the ground state with chemical potential $\mu = B_0$) has been studied in \cite{Loch_2005,Rodr_guez_2010,Rodr_guez_2009} with some additional assumptions on the region $\Lambda$. The case of $\mu=B_0$ has been solved by Charles and Estienne in \cite{Charles2020}, and then the case of an arbitrary $\mu\ge B_0$ by Leschke, Sobolev and Spitzer in \cite{Leschke2021}, both  under some regularity assumptions on the boundary $\partial \Lambda$. Our main result is \autoref{maincor}. It shows that the leading order asymptotic growth of the entanglement entropy for arbitrary $\alpha>0$ does not change, if we add such a slight perturbation in both the magnetic field and the electric potential, assuming some differentiability of these perturbations in the case $\alpha \le \frac 1 3$, depending on $\alpha$. Hence, we will not need to recalculate the value of the leading term, as we only estimate that this perturbation leads to an error term of smaller order in the scaling parameter $L$.

Our proof is based on a statement by Aleksandrov and Peller in \cite{Peller2010}, which is \autoref{peller} in this paper. With the help of this and approximations of the R\' enyi entropy functions $h_\alpha$ (see \eqref{h_a}), we can reduce our result to some $p$-Schatten (quasi-)norm estimates, as we prove in \autoref{main result section}.

Proving these $p$-Schatten (quasi-)norm estimates relies on the fact that some Sobolev embeddings on bounded subset of $\mathbb R^d$ are in some $p$-Schatten classes, which we specify and prove in \autoref{Sob emb Schatten2}. It is based on a result by Gramsch in \cite{Gramsch1968}. This allows us to estimate the $p$-Schatten (quasi-)norms of operators with sufficiently differentiable kernels. To get a representation of the kernel of the spectral projection of the perturbed Hamiltonian, we use the contour integral representation and the resolvent expansion. This has recently been done for perturbations of the free case ($H=-\nabla^2+V$) by  Müller and Schulte in \cite{muller2020stability}, which inspired me to try this approach. In our case ($B_0 >0$), we use an expanded resolvent expansion. The discrete spectrum allows us to explicitly resolve the contour integral for most terms. The general idea is explained in \autoref{Ansatz}, while the required kernel estimates are proven in the remaining sections.

The magnetic case (with an asymptotically constant magnetic field) appears simpler and more stable than the free case with the (negative) Laplacian as its single-particle Hamiltonian. From a technical point of view this is due to the gaps in the purely essential spectrum and the exponential decay of eigenfunctions of the Landau Hamiltonian. This is also the reason for an area law growth (without any logarithmic enhancement as in the free case), see also \cite{PasturSlavin14}.

\section*{Acknowledgment}
I would like to thank Wolfgang Spitzer for introducing me to this topic, proof reading multiple previous versions, and generally providing advice.

\section{Notations and preliminaries}\label{notations}
Let $\mathbb N=\{0,1,2, \dots\}$ be the natural numbers and $\mathbb Z^+$ be the positive integers.

Let $n,d$ be positive integers and $k$ be a natural number. For $x \in \mathbb R^n$ or $x \in \mathbb C^n$, let $\lVert x \rVert$ be its $2$-norm. 
The space of $p$-integrable (respectively essentially bounded if $p=\infty$), complex valued functions on $\mathbb R^n$ is called $\Lp^p(\mathbb R^n)$. The Sobolev space $W^{k,p}(\mathbb R^n)$ is the subspace of $\Lp^p(\mathbb R^n)$, such that their first $k$ distributional derivatives in any combination of directions are in $\Lp^p(\mathbb R^n)$. We define $C^k_b(\mathbb R^n,\mathbb C^d)$ as the subspace of $C^k(\mathbb R^n,\mathbb C^d)$, such that all derivatives of order $0 \le j \le n$ are bounded.

For any non-empty set $\Lambda \subset \mathbb R^n$ and any point $x \in \mathbb R^n$, we define the distance as
\begin{align}
\operatorname{dist}(x,\Lambda) \coloneqq \inf _{y \in \Lambda } \lVert x-y \rVert ,
\end{align}
and for any $r>0$ we define the $r$-neighbourhood of $\Lambda$ as
\begin{align}
D_r(\Lambda) \coloneqq \{ y \in \mathbb R^n \vert \operatorname{dist} (y,\Lambda)<r \} .
\end{align}
Furthermore, $1_\Lambda \colon \mathbb R^n \to \{0,1\} \subset \mathbb R$ is the indicator function of $\Lambda$, $\Lambda^\complement \coloneqq \mathbb R^n \setminus \Lambda$ is the complement of $\Lambda$, and
 if $\Lambda$ is measurable, let $\lvert \Lambda \rvert$ be its $n$-dimensional Lebesgue measure. If $\Lambda$ has Lipschitz-boundary $\partial\Lambda$, let $\lvert \partial \Lambda \rvert$ be the $(n-1)$-dimensional Hausdorff measure of $\partial \Lambda$. 

For any $x \in \mathbb R^n$, we define the disk $D_r(x)=D_r(\{x\})$. For any $x \in \mathbb C^n, j \in \mathbb N$, we inductively define $x^{\otimes j} \in  \left(\mathbb C^n\right)^{\otimes j} \cong \mathbb C^{n^j}$ by setting $x^{\otimes 0}\coloneqq 1 \in \mathbb C  \eqqcolon \left( \mathbb C^n\right) ^{\otimes 0} $ and $x^{\otimes (j+1)} \coloneqq x^{\otimes j} \otimes x \in \left( \mathbb C^n \right)^{\otimes j} \otimes \mathbb C^n \eqqcolon \left( \mathbb C^n \right)^{\otimes(j+1)} \cong  \mathbb C^{n^{j+1}}$. Every appearance of $\blank^{\otimes j} $  refers to this tensor product.

By $J$ we denote the matrix
\begin{align}
J \coloneqq \begin{pmatrix} 0 &1 \\
-1 &0
\end{pmatrix}. \label{def J}
\end{align}

For a complex number $\zeta$, let $\Re \zeta$ be its real part. 

For a multiplication operator with a function $G \colon \mathbb R^2 \to \mathbb C^n$, we use a slight abuse of notation and call it $G$ as well. This is relevant to decide, whether we are applying an operator to the underlying function or taking the composition of a multiplication operator and any other operator. Whenever there are both multiplication operators and other operators present in an expression, we regard $G$ as the multiplication operator, unless we write $G({\cdot})$. 

$C$ will always refer to a generic constant, that may depend on some, but never on all variables. $F$ will be used similarly, but the dependency on one complex variable will be important, which is why we write $F$ as a function of that variable. Both may change from line to line.

For any compact operator $S$ and any $p \in \mathbb R^+$, we define the $p$-Schatten von Neumann (quasi-)norm by the expression
\begin{align}
\lVert S \rVert_p^p \coloneqq \sum_{n \in \mathbb Z^+} s_n(S)^p,
\end{align}
where $(s_n(S))_{n \in \mathbb Z^+}$ is the decreasing sequence of singular values of $S$ counted with multiplicity. The operator norm of $S$ is written as $\lVert S \rVert_\infty$. We say an operator is in the $p$-Schatten class, if its $p$-Schatten norm is finite. For any pair of Hilbert spaces $\mathcal H_1,\mathcal H_2$, let $S_p(\mathcal H_1, \mathcal H_2)$ be the (quasi-)normed space of all $p$-Schatten class operators from $\mathcal H_1$ to $\mathcal H_2$. 

We recall some properties of the $p$-Schatten von Neumann (quasi-) norms. In the following, we will refer to them as $p$-Schatten norms. 
\begin{proposition} \label{pnorm facts}
Let $0 <p \le q\le \infty$ and let $S,T$ be operators on a Hilbert space.
The $p$-Schatten norm satisfies the properties
\begin{description}
\item[Monotonicity I] $\lVert S \rVert_p \ge \lVert S \rVert_q$,
\item[Monotonicity II] If $S \ge T \ge 0$, then $\lVert S \rVert_p \ge \lVert T \rVert_p$,
\item[Triangle inequality] If $p \ge 1$, then $\lVert S+T \rVert_p \le \lVert S \rVert_p + \lVert T \rVert_p$,
\item[$p$-triangle inequality]  If $p \le 1$, then $\lVert S+T \rVert_p^p \le \lVert S \rVert_p^p + \lVert T \rVert_p^p$,
\item[Powers] If $S \ge 0$, then $\lVert S^p \rVert_q^q = \lVert S^q \rVert_p^p$,
\item[Square] $\lVert S \rVert_p ^2= \lVert S^*S \rVert _{p/2}$, where $S^*$ denotes the adjoint of $S$,
\item[Adjoint] $\lVert S^*\rVert_p = \lVert S \rVert _p$.
\item[Hölder I] Let $\frac 1 r= \frac 1 p + \frac 1 q $. Then $\lVert ST \rVert_r \le \lVert S \rVert _p \lVert T \rVert_q$.
\item[Hölder II] Let $\frac 1 r= \frac \alpha p + \frac {1-\alpha} q $ with $0 < \alpha <1$. Then $\lVert S \rVert_r \le \lVert S \rVert _p^\alpha \lVert S \rVert_q^{1-\alpha}$.
\item[Hilbert--Schmidt kernel] If $T \colon \Lp^2(\mathbb R^{d_1}) \to \Lp^2 (\mathbb R^{d_2})$ has  an integral kernel $t$, which is square integrable, then $\lVert T \rVert_2 = \lVert t \rVert_{\Lp^2(\mathbb R^{d_1+d_2})}$.
\item[Orthogonality] If $S T^*=0$ or $S^*T=0$, then $ \lVert S \rVert_p \le \lVert S+T \rVert_p $.
\end{description}
\end{proposition}
Most of these have for example been proven by McCarthy in \cite{McCarthy1967}. We will now briefly prove the remaining ones.
\begin{proof}
``Monotonicity II'' follows, as the inequality holds for the ordered sequence of singular values.  ``Hölder II'' is an application of ``Hölder I" with the operators $\lvert S \rvert^\alpha$ and $\lvert S \rvert ^{1-\alpha}$ and the properties  ``Square'' and ``Powers''. ``Hilbert--Schmidt kernel'' can be seen as a corollary of Lemma 2.2 in \cite{McCarthy1967}. ``Orthogonality'' is based on the observation, that if $S^*T=0$, we have $(S+T)^*(S+T)=S^*S+T^*T$, ``Monotonicity II'', and ``Adjoint'' to replace the condition $S^*T=0$ by the non-equivalent condition $ST^*=0$. 
\end{proof}

\begin{definition}
We say a densely defined operator $T$ on $\Lp^2(\mathbb R^2)$ has the integral kernel $t \colon \mathbb R^2 \times \mathbb R^2 \to \mathbb C$, if for any $f \in C^0_c(\mathbb R^2)$, the identity
\begin{align}
(Tf)(x)= \int_{\mathbb R^2} t(x,y) f(y) dy
\end{align}
holds for almost all $x \in \mathbb R^2$. In this case, we define
\begin{align}
\operatorname{iker}T(x,y) \coloneqq t(x,y).
\end{align} 
We say, that $t$ is \emph{nice}, or respectively, that $T$ is a \emph{nice} integral operator, if for any fixed $x$, the functions $t(x, \blank)$ and $t(\blank, x)$ are in $\Lp^1(\mathbb R^2)$ with a norm bounded independently of $x$. 
\end{definition}

\begin{corollary} \label{nice>bounded}
Let $T$ be a \emph{nice} integral operator. Then $T$ is a bounded operator on $\Lp^2(\mathbb R^2)$. 
\end{corollary}
\begin{proof}
The expression $\left \lVert \left( x \mapsto \lVert t(x,{\cdot}) \rVert_{\Lp^1(\mathbb R^2)} \right) \right \rVert_{\Lp^\infty(\mathbb R^2) }$ is finite and an upper bound for the operator norm of  $T$ as  an operator on $\Lp^\infty(\mathbb R^2)$. On the other hand, the expression $\left \lVert \left( y \mapsto \lVert t({\cdot},y) \rVert_{\Lp^1(\mathbb R^2)} \right) \right \rVert_{\Lp^\infty(\mathbb R^2) }$ is finite and an upper bound for the operator norm of $T$ as a bounded operator on $\Lp^1(\mathbb R^2)$. Hence, by the Riesz-Thorin interpolation theorem, the operator $T$ is bounded on $\Lp^2(\mathbb R^2)$ with an operator norm bounded by the square root of the product of both of these expressions.
\end{proof}
\begin{lemma}
Let $S,T$ be \emph{nice} integral operators on $\Lp^2(\mathbb R^2)$ with integral kernels $s,t$. Let $x,z \in \mathbb R^2$. Then we have the identities
\begin{align}
\operatorname{iker} (S+T)(x,z)=& (s+t)(x,z), \\
\operatorname{iker} (ST)(x,z) =& \int_{\mathbb R^2} s(x,y) t(y,z) dy .
\end{align}
In particular, $S+T$ and $ST$ are \emph{nice} integral operators.
\end{lemma}
The first statement is trivial and the second follows by Fubini to interchange the integral over $y$ with the one over $z$, for any test function $f \in \Lp^1(\mathbb R^2) \cap \Lp^\infty(\mathbb R^2)$. 
\begin{definition}
Let $\gamma ,d \in \mathbb N$ and $\lambda \in \mathbb [0, \infty)$. Then we define the space $W^{\gamma, \infty}_{(\lambda)}(\mathbb R^2, \mathbb C^d)$ as the subspace of the Sobolev space  $W^{\gamma, \infty}(\mathbb R^2, \mathbb C^d)$, where the norm
\begin{align}
\lVert u \rVert_{W^{\gamma, \infty}_{(\lambda)}(\mathbb R^2, \mathbb C^d)}  \coloneqq \sum_{\gamma' \le \gamma } \sup_{x \in \mathbb R^2} \left \lVert \left( 1 +\lVert x \rVert \right)^\lambda \left( \nabla^{\otimes \gamma'} u\right) (x)  \right \rVert
\end{align}
is finite. The supremum in this definition refers to the almost everywhere  supremum. This is a Banach space. The limit space
\begin{align}
W^{\gamma, \infty}_{(\infty)}(\mathbb R^2, \mathbb C^d) \coloneqq \bigcap_{\lambda \ge 0 } W^{\gamma, \infty}_{(\lambda)}(\mathbb R^2, \mathbb C^d)
\end{align}
is only a  vector space equipped with the inverse limit topology associated to the intersection (a set is open, if and only if it is open in each space for finite $\lambda$.).
\end{definition}
These spaces are motivated by Schwartz semi-norms.

\section{Setting and main result}\label{main result section}

We introduce the Landau Hamilton operator $H_0$ with a constant magnetic field $B_0 >0$, defined on (a suitable subspace of) $\Lp^2(\mathbb R^2)$, with magnetic gauge $A_0$ given by
\begin{align}
A_0(x) \coloneqq &\frac {B_0} 2 Jx, \label{def A0} \\ 
H_0 \coloneqq& (-i \nabla - A_0 )^2. \label{def H0}
\end{align}
The spectrum of $H_0$, $\sigma(H_0)$, equals $B_0(2 \mathbb N+1)$. Let $P_l$ be the projection onto the eigenspace with eigenvalue $B_0(2l+1)$ for $l \in \mathbb N$.

Furthermore, for any $\alpha>0$, we introduce the $\alpha$-R\' enyi entropy functions $h_\alpha \colon [0,1] \to [0, \ln (2)]$,
\begin{align}
h_\alpha(x) \coloneqq \begin{cases} \frac 1 {1-\alpha} \ln\left( x^\alpha+ (1-x)^\alpha \right) & \text{ for } \alpha \not=1 ,\\
-x\ln x - (1-x) \ln(1-x) & \text{ for } \alpha=1, \end{cases} \label{h_a}
\end{align}
for $x \in (0,1)$ and $h_\alpha(0)=h_\alpha(1)=0$. Throughout this paper, let $\Lambda \subset \mathbb R^2$ be a bounded open set with Lipschitz-boundary. 

Let $\mu \in \mathbb R \setminus B_0(2\mathbb N+1)$. We define $1_{\le\mu}(H_0)$ as the spectral projection associated to $H_0$ and $\mu$. 
We are interested in how the leading order asymptotic expansion of the local entropy,
\begin{align}
S_\alpha(L \Lambda)  \coloneqq\operatorname{tr} h_\alpha \left(1_{L\Lambda} 1_{\le \mu}(H_0) 1_{L\Lambda } \right) ,
\end{align}
as $L \to \infty$ changes under slight perturbations of $H_0$. The trace is defined as the usual Hilbert space trace of trace class operators on $\Lp^2(\mathbb R^2)$.  This quantity is the local entropy or entanglement entropy of the ground state restricted to $L \Lambda$.  Under the assumption that $\Lambda$ has $C^3$ boundary, the leading term of order $L$ for  the operator $H=H_0$ has been calculated by Leschke, Sobolev and Spitzer in \cite{Leschke2021}. This allows us to focus on bounding the error term that arises, as we introduce a perturbation to $H_0$. Our main result is \autoref{maincor} and relies on the exact calculations of the leading term for $H=H_0$, see \cite{Leschke2021}, and the estimates we will prove in this paper. 

 The following condition is needed to state our main results and a lot of results along the way. Throughout this paper, we fix $0<\varepsilon<1$.
\begin{definition}
Let $\gamma$ be a natural number. We call a magnetic field $B_\varepsilon  \colon \mathbb R^2 \to \mathbb R$ and a potential $V_\varepsilon\colon \mathbb R^2 \to \mathbb R$  $(\gamma,\varepsilon)$ tame, if $V_\varepsilon \in W^{\gamma, \infty} _{(\varepsilon)}(\mathbb R^2, \mathbb R)$ and $B_\varepsilon \in W^{\gamma, \infty} _{(1+\varepsilon)}(\mathbb R^2, \mathbb R)$.
\end{definition}
\begin{remark}
All the following estimates will depend on $B_\varepsilon, V_\varepsilon$ only through $ \varepsilon, \gamma$ and the norms of $B_\varepsilon,V_\varepsilon$ in the spaces $W^{\gamma, \infty} _{(1+\varepsilon)}(\mathbb R^2, \mathbb R)$ and $W^{\gamma, \infty} _{(\varepsilon)}(\mathbb R^2, \mathbb R)$. Maybe somewhat counter-intuitively, small values of $\varepsilon$ correspond to slowly decaying $B_\varepsilon, V_\varepsilon$.
\end{remark}
To define the perturbed Hamiltonian $H$, we need to choose a gauge $A_\varepsilon$ of the magnetic field $B_\varepsilon$. We choose the convolution, which is given by
\begin{align}
A_\varepsilon(x) \coloneqq\left( B_\varepsilon * \frac{J\blank}{2\pi \lVert \blank \rVert^2} \right)(x)= \int_{\mathbb R^2} B_\varepsilon(x-y) \frac{Jy}{2\pi \lVert y \rVert^2} dy \label{def A eps eq}
\end{align}
for any $x \in \mathbb R^2$.
Its relevant properties are summed up in the following Lemma.
\begin{lemma} \label{dA def} Let $\gamma  \in \mathbb N$, $f \in W_{(1+\varepsilon)}^{\gamma,\infty}(\mathbb R^2 , \mathbb R)$ and define $g \in \operatorname{Maps}(\mathbb R^2, \mathbb R^2)$ as the convolution
\begin{align}
g \coloneqq f * \frac{J\blank}{2\pi \lVert \blank \rVert^2}.
\end{align}
 Then, for any $x \in \mathbb R^2$, we have the identities
\begin{align}
\nabla_x \times g(x) &=f(x) ,\\
\nabla_x \cdot g(x) &= 0 \label{divA=0}.
\end{align}
Furthermore, we have $g \in W^{\gamma, \infty} _{(\varepsilon)}(\mathbb R^2, \mathbb R^2)$.
\end{lemma}

\begin{remark}
A gauge satisfying \eqref{divA=0} is commonly referred to as a Coulomb gauge.
The restriction to $\varepsilon<1$ is necessary to get the described decay. A value of $\varepsilon>1$ will only achieve a $(1+\lVert x \rVert)^{-1}$ decay in $A_\varepsilon$. 
\end{remark}

The proof can be found in \autoref{Appendix B}. 

 Now we define the perturbed gauge $A$ and the perturbed Hamiltonian $H$ by
\begin{align} 
A &\coloneqq A_0-A_\varepsilon ,\\
H&\coloneqq (- i \nabla -A)^2+V_\varepsilon. \label{def H}
\end{align}
As we can see, this gauge corresponds to the magnetic field ${B_0} - B_\varepsilon$, that is, $\nabla_x \times A(x)= B_0 -B_\varepsilon(x)$. The operator $H$ is self-adjoint and its domain agrees with the domain of $H_0$, which we will see in \autoref{ess spectrum}.

We need the following $p$-Schatten quasi-norm estimate, which will be proven in the next section. 

\begin{theorem} \label{p-norm estimate final}
Let $l \in \mathbb N,\gamma \in \mathbb Z^+$. Let $B_\varepsilon,V_\varepsilon$ be $(\gamma,\varepsilon)$ tame and let $1 \ge p> \frac 2 {\gamma +3 }  $. Let $a ,b \in \mathbb R \setminus B_0(2\mathbb N+1)$ with $a<b$. Then we have the estimates
\begin{align}
\left \lVert 1_{L\Lambda} 1_{[a,b]} (H) 1_{L \Lambda^\complement} \right \rVert_p^p& \le CL, \label{pnef 1}\\
\left \lVert 1_{L\Lambda} \left(1_{[a,b]} (H) -1_{[a,b]}(H_0) \right) 1_{L\Lambda^\complement} \right \rVert_p^p& \le CL^{1-p \varepsilon}. \label{pnef 2}
\end{align}
The constants $C$ depend on $\gamma, a,b, \Lambda, p ,\varepsilon,B_\varepsilon,V_\varepsilon$.
\end{theorem}

Finally, we need the following statement due to Aleksandrov and Peller, which is a Corollary of Theorem 5.11 in \cite{Peller2010} and the inclusion $C_c^\infty(\mathbb R) \subset B_{\infty,1}^1(\mathbb R)$, where the latter refers to the Besov space as used by Aleksandrov and Peller.
\begin{proposition} [based on Theorem 5.11 in \cite{Peller2010}]\label{peller}
Let $f \in C^\infty_c(\mathbb R)$. Then there is a constant $C<\infty$, such that for any self-adjoint bounded operators $A,B$, such that $A-B$ is trace class, we have the estimate
\begin{align}
\lVert f(A)-f(B) \rVert_1 \le C\lVert A-B \rVert_1 \, .
\end{align}
\end{proposition} 
Now we state the key result of this paper, which is proved below.
\begin{theorem} \label{main}
Let $\alpha>0$ and choose $\beta= \min( 0.5,\alpha)$. Define $\gamma $ as the smallest  integer, such that $\gamma >  \frac 1 \beta -3$. Let $B_\varepsilon,V_\varepsilon$ be $(\gamma,\varepsilon)$ tame. Let  $a ,b \in \mathbb R \setminus {B_0}(2 \mathbb N+1),a<b$ and $I\coloneqq [a,b]$. Then we have
\begin{align}
\operatorname{tr}\left(   h_\alpha \left(1_{L \Lambda }1_{I} (H)1_{L\Lambda}\right) - h_\alpha \left(1_{L\Lambda}1_I(H_0)1_{L \Lambda} \right)\right)   =  o (L)  \label{result formula},
\end{align}
as $L \to \infty$.
\end{theorem}

\begin{remark}
The choice of $\beta=0.5$ for $\alpha\ge 0.5$ delivers the optimal value for $\gamma$, namely $0$. For $\alpha>\frac 1 3$, we can get away with a non-differentiable $B_\varepsilon,V_\varepsilon$. 

The assumption that $a,b \not \in B_0(2\mathbb N+1)$ cannot be dropped, as the following counter example illustrates. Let $B_\varepsilon=0, a=0$ and $b=B_0$. By \autoref{ess spectrum}, the spectrum of $H$ has an accumulation point at $B_0$. If we assume $V_\varepsilon>0$ pointwise, then all eigenvalues of $H$ are strictly larger than $B_0$ and hence $1_I(H)=0$. But Theorem 8 in \cite{Leschke2021}, which we will elaborate on shortly, states, that the leading order asymptotic expansion of $\operatorname{tr} h_\alpha \left(1_{L\Lambda}1_I(H_0)1_{L \Lambda} \right)$ for large $L$ is of order $O(L)$ and does not vanish. On the other hand, if we assume that $-B_0<V_\varepsilon<0$ pointwise, there is a spectral gap of the form $(B_0,2B_0)$ in the spectrum of $H$. Hence, we can move $b$ to $1.5B_0$ without changing the operators. Now we can apply our \autoref{main}. Hence under our general assumptions, it is possible to get both one-sided limits, when $b=B_0$. We expect similar results, whenever $a$ or $b$ are in the spectrum of $H_0$. It is, however, a little more complicated to see, whether the leading order expansion for $H_0$ changes, when we add or remove a single Landau level from the interval $I$.
\end{remark}
The following corollary is our main result. It combines Theorem 8 in \cite{Leschke2021}, which can be stated as the corollary for the case $B_\varepsilon=V_\varepsilon=0$, with our \autoref{main}.
\begin{corollary} \label{maincor}
Let $\alpha>0$ and choose $\beta=\min( \alpha, 0.5)$. Define $\gamma $ as the smallest positive integer, such that $\gamma >  \frac 1 \beta -3$. Let $B_\varepsilon,V_\varepsilon$ be $(\gamma,\varepsilon)$ tame. Let $\mu \not \in \sigma(H_0)$ and define $\nu$ as the largest integer, such that $B_0(2\nu+1)<\mu$. Assume that the boundary $\partial \Lambda$ is $C^3$-smooth.  Then 
\begin{align}
S_\alpha(L \Lambda) = \operatorname{tr}(h_\alpha(1_{L\Lambda}1_{\le \mu} (H)1_{L\Lambda} ) = L \sqrt{B_0} \lvert \partial \Lambda \rvert  M_{\le \nu}(h_\alpha) + o(L),
\end{align}
as $L \to \infty$ with $0< M_{\le \nu}(h_\alpha)<\infty$ as described in \cite{Leschke2021} for $\nu \ge 0$ and $M_{\le \nu}(h_\alpha) \coloneqq 0$ for $\nu<0$.
\end{corollary}
In the case $\mu<B_0$, the projection is finite dimensional and the entropy has an order at most $O(1)$ in $L$ as $L \to \infty$.

\begin{proof}[Proof of \autoref{main}]
We define the function $g_\alpha \colon [0,1] \to [0,\ln(2)]$ by the identity
\begin{align}
g_\alpha(4x(1-x))=h_\alpha(x).
\end{align}
The symmetry of $h_\alpha$ guarantees the existence of $g_\alpha$. We have
\begin{align}
g_\alpha(t)=h_\alpha \left(\frac{1- \sqrt{1 -t}}2 \right) .
\end{align}
Let $\varepsilon_0>0$. We choose a smooth cut-off function  $\varphi \colon [0,1] \to [0,1] $ with $\varphi(x) = 1$, if $ x \le \varepsilon_0$, and $\varphi(x) = 0,$ if  $ x \ge 2\varepsilon_0$. 
Now we write
\begin{align}
g_\alpha(t)= (1-\varphi(t)) g_\alpha(t) +\varphi(t)g_\alpha(t).
\end{align}
The advantage of this decomposition is that the first summand is smooth, and the second summand is small. The second summand can be bounded using the fact, that $h_\alpha$ is $\beta$-Hölder continuous on $[0,1]$ and smooth on $(0,1)$. As $h_\alpha$ is symmetric around $t=\frac 12$ and analytic on $(0,1)$, its Taylor expansion at that point contains only even powers of $(t-\frac 1 2)$. Thus, we see that $g_\alpha$ is analytic at $t=1$. Hence $g_\alpha \in C^\infty((0,1])$ and it is $\beta$-Hölder continuous on $[0,1]$, as $\beta= \min(\alpha,0.5)$.

We choose $\beta' < \beta \le  \frac 1 2 $, such that $\gamma > \frac 1 {\beta'} -3$. 
Hence, we have
\begin{align}
\varphi(t) g_\alpha(t) \le C \varepsilon_0^{\beta- \beta'} t^{\beta'}. \label{phi g}
\end{align}
We define $P,P'$ as the spectral projections,
\begin{align}
P \coloneqq& 1_I(H_0), \\
P'\coloneqq& 1_I(H). 
\end{align}
We observe
\begin{align}
h_\alpha(1_{L\Lambda} P^{(')}1_{L \Lambda} )= g_\alpha ( 4 \lvert 1_{L \Lambda^\complement} P^{(')} 1_{L\Lambda} \rvert^2 ).
\end{align}
We can now apply \autoref{peller}. Thus,
\begin{align}
&\left \lVert \left((1-\varphi)g_\alpha\right)\left(4\lvert 1_{L\Lambda^\complement} P' 1_{L \Lambda} \rvert ^2 \right) - \left((1-\varphi)g_\alpha\right)\left(4\lvert 1_{L\Lambda^\complement} P 1_{L \Lambda} \rvert ^2 \right)  \right \rVert_1\\
\le& C \left \lVert \lvert 1_{L\Lambda^\complement} P' 1_{L \Lambda} \rvert ^2 - \lvert 1_{L\Lambda^\complement} P 1_{L \Lambda} \rvert ^2 \right \rVert_1\\
\le & C \left \lVert  1_{L\Lambda^\complement} (P'-P) 1_{L \Lambda}\right \rVert_1 \\
\le &CL^{1- \varepsilon } \, .
\end{align}
Note that the last constant $C$ depends on $\varepsilon_0$, but not on $L$. In the second step we used the identity $\lvert A\rvert^2- \lvert B \rvert^2=A^*(A-B)+(A^*-B^*)B$.  
In the last step, we used \autoref{p-norm estimate final} with $p=1$.

We can also apply \autoref{p-norm estimate final} for the remaining term, after using \eqref{phi g}, $1 \ge 2 \beta' > \frac 2 {\gamma+3}$ and that $H=H_0$ is admissible for \autoref{p-norm estimate final}.
\begin{align}
\left\lVert \left(\varphi g_\alpha\right) \left(4 \lvert 1_{L\Lambda^\complement} P^{(')} 1_{L \Lambda} \rvert ^2 \right)\right \rVert_1 
\le   C \varepsilon_0^{\beta- \beta'} \left \lVert \lvert 1_{L\Lambda^\complement} P^{(')} 1_{L \Lambda} \rvert \right \rVert_{2\beta'}^{2\beta'}
\le C\varepsilon_0^{\beta-\beta'}L.
\end{align}
Hence,
\begin{align}
\left \lvert \operatorname{tr} h_\alpha (1_{L\Lambda}P'1_{L\Lambda})-\operatorname{tr}h_\alpha (1_{L\Lambda}P1_{L\Lambda}) \right \rvert \le C(\varepsilon_0)L^{1-\varepsilon}+C\varepsilon_0^{\beta-\beta'}L.
\end{align}
Note that the first constant $C(\varepsilon_0)$ depends on $\varepsilon_0$ while the second one does not. This term is in $o(L)$, as for any $\varepsilon>0$ we can choose $L$ large enough to let the first term be less than $\varepsilon_0L$. This proves that the leading term expansion of the $\alpha$-R\' enyi entropy for the perturbed Landau Hamiltonian $H$ agrees with the main term in the same expansion for the Landau Hamiltonian $H_0$. This finishes the proof.
\end{proof}
\begin{remark}
We can actually pick $\varepsilon_0$ dependent on $L$, which does lead to a smaller error term, if we bound the constant $C(\varepsilon_0)$ more precisely. This does however not lead to an improved error term in  \autoref{maincor}, as the known error term for the constant magnetic field is too large. Hence, I did not include the details here.
\end{remark}

\section{The Ansatz for the proof of \autoref{p-norm estimate final}}\label{Ansatz}

The goal of this section is to explain how to prove \autoref{p-norm estimate final} and, to reduce it to two more technical statements. The general approach has been inspired by \cite{muller2020stability}.

We define
\begin{align}
H_\varepsilon&\coloneqq H-H_0, \label{def Heps}
\end{align}
where $H$ and $H_0$ were defined in \eqref{def H} and \eqref{def H0}.

 We expand $H_\varepsilon$ as 
\begin{align}
H_\varepsilon&= H-H_0 \\
&= \left( - i \nabla -A \right)^2 - \left( - i \nabla -A_0 \right)^2 +V_\varepsilon\\
&= (A_0-A) \cdot \left( - i \nabla -A +A_0 - A_0  \right) + \left( - i \nabla -A_0 \right)\cdot  (A_0-A) +V_\varepsilon\\
&=  2A_\varepsilon  \cdot \left( - i \nabla  - A_0  \right)  + A_\varepsilon^2 +V_\varepsilon.
\end{align}
We used the identity $a^2-b^2=(a-b)a+b(a-b)$ in the third step and \eqref{divA=0}, which is equivalent to $\nabla \cdot  A_\varepsilon=A_\varepsilon \cdot \nabla$, in the last step. We now introduce the pseudo potential
\begin{align}
W_\varepsilon \coloneqq A_\varepsilon^2+ V_\varepsilon.
\end{align}

We introduce a few more operators. 
Let $I \subset  \mathbb N$ be cofinite, $  \zeta  \in \mathbb C$ and $\zeta \not \in B_0(2I+1)$. Then we define the bounded operator
\begin{align}
M_{I,\zeta} \coloneqq \sum_{l \in I} \frac{P_l}{B_0(2l+1)-\zeta}. \label{def M_Iz}
\end{align}
It satisfies $M_{I,\zeta}^*=M_{I, \overline{\zeta}}$. For $\zeta \not \in \sigma(H_0)$, we have the identity
\begin{align}
M_{\mathbb N,\zeta}= \frac 1 {H_0-\zeta}.
\end{align}
There are some results describing the kernel of the resolvent operator, but we also need the special case   
\begin{align}
T_l \coloneqq M_{\mathbb N \setminus \{ l \} , {B_0}(2l+1)} = \sum_{k \not = l } \frac{P_k}{2{B_0}(k-l)}. \label{def T_l}
\end{align}
Hence it is more convenient to deal with the operator $M_{I,\zeta}$ in this generality.

We define $n_0$ as the smallest integer such that
\begin{align}
n_0 > \frac 1 {2\varepsilon} \label{def n0}.
\end{align}

The following lemma will be proved in  \autoref{last section} after some preparations.
\begin{lemma} \label{4-Schatten}
Let $B_\varepsilon$ and $ V_\varepsilon$ be $(0,\varepsilon)$ tame. Then for any $I \subset N$ cofinite and any $\zeta \in \mathbb C \setminus {B_0}(2I+1)$, the operator $H_\varepsilon M_{I,\zeta}$ is in the $4n_0$-Schatten class, and the $4n_0$-Schatten norm is in $\Lp^\infty_{loc}(\mathbb C \setminus \sigma(H_0))$ as a function of $\zeta$. The upper bound for the norm depends on ${B_0}$.
\end{lemma}
As $p$-Schatten class operators are compact, we now know that $H_\varepsilon$ is relatively $H_0$-compact. This implies 
\begin{corollary} \label{ess spectrum}
The essential spectrum of $H$ agrees with the essential spectrum of $H_0$ which is ${B_0}(2\mathbb N+1)$.
\end{corollary}
\begin{remark}
The statement is also true if $V=0$ and $B$ is smooth and converges to $B_0$ as $\lVert x \rVert \to \infty$ (at any rate), see \cite{iwatsuka1983}. They state smoothness of $B$ as a condition, but I think it is not required. However, their algebraic proof does not imply that the eigenspaces of $H_0$ and $H$ are at all related.
\end{remark}

As $\sigma(H_0)$ is discrete, this implies, that $\sigma (H) = \overline{\sigma_p (H)} $ and that the continuous part of the spectrum of $H$ vanishes. 
We continue with the Riesz integral representation.
\begin{fact} \label{path integral operator}
For any path $\Gamma$ in $\mathbb C$ that intersects $\mathbb R$ in exactly two points $\lambda_1<\lambda_2$, does not intersect $\sigma(H) \subset \mathbb R$ and has winding number $+1$ around $(\lambda_1+\lambda_2)/2$, we have the identity
\begin{align}
-\frac 1 {2\pi i}\int_\Gamma \frac{ d\zeta} {H-\zeta} = 1_{\lambda_1<E<\lambda_2}(H).
\end{align}
\end{fact}

With the resolvent identity, we can write
\begin{align}
\frac 1 {H-\zeta} =& \frac  1 {H_0-\zeta} - \frac 1 {H-\zeta} H_\varepsilon \frac 1 {H_0-\zeta} \\
=& \frac  1 {H_0-\zeta} - \frac 1 {H_0-\zeta} H_\varepsilon \frac 1 {H_0-\zeta} + \frac 1 {H_0-\zeta} H_\varepsilon\frac 1 {H-\zeta} H_\varepsilon \frac 1 {H_0-\zeta}\, .
\end{align}

By induction, this leads to 
\begin{corollary} \label{iterated resolvent identity}
For any $n \in \mathbb Z^+, \zeta \not \in \sigma(H)\cup \sigma(H_0)$, we have 
\begin{align}
\frac 1 {H-\zeta} = \sum_{k=0}^{2n-1} \frac {(-1)^k} {H_0-\zeta} \left( H_\varepsilon  \frac 1 {H_0-\zeta} \right)^k + \left(   \frac 1 {H_0-\zeta}H_\varepsilon \right)^n \frac 1 {H-\zeta} \left( H_\varepsilon  \frac 1 {H_0-\zeta} \right)^n,
\end{align}
where $H_\varepsilon=H-H_0$, as in \eqref{def Heps}.
\end{corollary}
\bigskip
For the summands in \autoref{iterated resolvent identity} except the last summand, we can resolve the path integral over some paths.
\begin{lemma} \label{path integral resolve}
Let $l,k \in \mathbb N$ and  $\Gamma$ be the path along the circle $\partial D_{B_0}({B_0}(2l+1))$ that rotates in positive direction. Then we have
\begin{align}
-\frac{1}{2\pi i}\int_\Gamma \frac 1 {H_0-\zeta} \left( H_\varepsilon  \frac 1 {H_0-\zeta} \right)^k d\zeta &= \sum_{m=0}^k (T_l H_\varepsilon)^m P_l (H_\varepsilon T_l)^{k-m} \label{iri eq},
\end{align}
where $H_\varepsilon=H-H_0$, as in \eqref{def Heps}.
\end{lemma}

\begin{proof} Let $N >2l$ and either $I=\mathbb N$ and $\zeta \in \Gamma$ or $I=\mathbb N \setminus \{l\}$ and $\zeta={B_0}(2l+1)$. We introduce $P_{\le N} \coloneqq\sum_{ n \le N} P_n$ and $P_{>N}\coloneqq 1-P_{\le N}$.
We continue with the identity 
\begin{align}
P_{\le N}M_{I,\zeta}= \sum_{j \in I, j \le N } \frac{P_{j}}{{B_0}(2j+1)-\zeta} \label{single resolvent}.
\end{align}
There is a constant $C$, independent of $N$ and $\zeta$, such that the estimate $\left \lVert {P_{>N}}  M_{I,\zeta} \right \rVert_\infty \le \frac C N$ holds (see \autoref{orth op fam}). Furthermore, by \autoref{4-Schatten} and as the $4n_0$-Schatten norm is an upper bound for the operator norm, we have the estimate $\left \lVert H_\varepsilon M_{I,\zeta}\right \rVert_\infty < C$ with a constant $C$ independent of $\zeta$ (and $N$). We use the telescope sum $b(ab)^k-c(ac)^k=\sum_{k'=0}^k (ba)^{k'}(b-c)  (ac)^{k-k'}$, which holds in any ring, and the triangle inequality to get 
\begin{align}
&\left \lVert M_{I,\zeta} \left( H_\varepsilon M_{I,\zeta}\right)^k- P_{\le N} M_{I,\zeta}\left(H_\varepsilon P_{\le N} M_{I, \zeta} \right)^k \right \rVert_\infty \\
\le & \sum_{k'=0}^{k} \left \lVert \left( M_{I,\zeta} H_\varepsilon \right)^{k'} P_{>N}M_{I,\zeta} \left( H_\varepsilon P_{\le N}M_{I,\zeta}  \right)^{k-k'} \right \rVert_\infty 
\le \frac {C} N \label{pir est},
\end{align}
where $C$ is independent of $N$ and $\zeta$. The second step relies on the submultiplicativity of the norm, and the identity $M_{I,\zeta}P_{\le N}=P_{\le N} M_{I,\zeta}$. Thus, we have 
\begin{align}
&-\frac{1}{2\pi i}\int_\Gamma \frac 1 {H_0-\zeta} \left( H_\varepsilon  \frac 1 {H_0-\zeta} \right)^k d\zeta \\
=&-\frac 1 {2\pi i} \lim_{N \to \infty} \int_\Gamma  P_{\le N} M_{\mathbb N,\zeta} \left( H_\varepsilon P_{\le N} M_{\mathbb N, \zeta} \right)^k d \zeta \\
=&-\frac 1 {2\pi i} \lim_{N \to \infty} \int_\Gamma  \sum_{\sigma \in \{0, \dots, N\}^{k+1}} \frac{P_{\sigma_0} \prod_{j=1}^k H_\varepsilon P_{\sigma_j}} {\prod_{j=0}^k ({B_0}(2\sigma_j+1)-\zeta)} d \zeta \\
=&\lim_{N \to \infty}\sum_{\sigma \in \{0, \dots, N\}^{k+1}}  P_{\sigma_0} \left(\prod_{j=1}^k H_\varepsilon P_{\sigma_j} \right)  \begin{cases}  \prod_{\sigma_j \not = l} \frac 1 {2{B_0}(\sigma_j-l)} &\mathrm {if  }\# \{ j \mid \sigma_j=l\}=1 \\ 0 & \mathrm{else} \end{cases} \\
=&\lim_{N\to \infty}  \sum_{m=0}^k (P_{\le N} T_l H_\varepsilon)^m P_l (H_\varepsilon P_{\le N} T_l)^{k-m} \\
=&\sum_{m=0}^k (T_l H_\varepsilon)^m P_l (H_\varepsilon T_l)^{k-m} .
\end{align}
In the first step, we used that \eqref{pir est} holds uniformly in $\zeta\in \Gamma$ for $I=\mathbb N$. 
In the second step, we inserted \eqref{single resolvent} $k+1$ times and multiplied out all terms in order to get a finite sum. We then exchanged this finite sum with the complex path integral and resolved this complex-valued integral. The fourth step uses \eqref{single resolvent} in reverse. 
The final step follows by \eqref{pir est} for $I=\mathbb N \setminus \{l\}$ and $\zeta={B_0}(2l+1)$. This finishes the proof.
\end{proof}

We will prove the following theorem at the end of \autoref{last section}. 
\begin{theorem} \label{cube norm}
Let $k,l,m, \gamma  \in \mathbb N$  with $k \ge m$. Let $B_\varepsilon,V_\varepsilon$ be $(\gamma,\varepsilon)$ tame and let $1 \ge p> \frac 2 {\gamma +3 }  $. Then there is a constant $C>0$ and a $\lambda >0$, such that for any $R\ge 0$, we have the upper bound 
\begin{align}
\left \lVert 1_{[0,1]^2+x_0} (T_l H_\varepsilon)^m P_l (H_\varepsilon T_l)^{k-m} 1_{D^\complement_R(x_0)}\right \rVert_p \le C\frac{ \exp \left(-\lambda R^2 \right)}{(1+\lVert x_0 \rVert)^{k \varepsilon}},
\end{align}
for any $x_0 \in \mathbb R^2$. 
The constant $C$ depends on $ {B_0},l,k,m,\gamma,p, \varepsilon,B_\varepsilon,V_\varepsilon$, but is independent of $R$ and $x_0$.
\end{theorem}
\begin{remark} For $k=m=0$, this is Lemma 12 in \cite{Leschke2021}. 
\end{remark}

We will now follow Theorem 13 in \cite{Leschke2021}. 
But we go a slightly different direction with the proof\footnote{We replace a sum by an integral.}.

\begin{theorem} \label{penultimate}
Let $k,l,m,\gamma \in \mathbb N$ with $k \ge m$, let $B_\varepsilon,V_\varepsilon$ be $(\gamma,\varepsilon)$ tame and let $1 \ge p> \frac 2 {\gamma +3 }  $. Then for any $L>1$ we have 
\begin{align}
\left \lVert 1_{L\Lambda}  (T_l H_\varepsilon)^m P_l (H_\varepsilon T_l)^{k-m} 1_{L\Lambda^\complement}\right \rVert_p^p \le CL^{1-p k\varepsilon}.
\end{align}
The constant $C$ depends on $\Lambda,{B_0},l,k,m,\gamma,p, \varepsilon, B_\varepsilon,V_\varepsilon$.
\end{theorem}

\begin{proof}
We define
\begin{align}
T \coloneqq (T_l H_\varepsilon)^m P_l (H_\varepsilon T_l)^{k-m}.
\end{align}

We choose an $h _0\in [0,1)^2$. We will now use the $p$-Schatten norm property we called orthogonality in the first and forth step, and the $p$-triangle inequality in the second step. Hence,
\begin{align}
&\left \lVert 1_{L\Lambda} T 1_{L\Lambda^\complement} \right \rVert_p^p \\
\le & \left  \lVert \sum_{z \in \mathbb Z^2, z +h_0\in D_{ \sqrt 2 }(L \Lambda )}  1_{[0,1)^2+z+h_0} T 1_{L\Lambda^\complement} \right \rVert_p^p \\
\le & \sum_{z \in \mathbb Z^2, z+h_0 \in D_{\sqrt 2 }(L \Lambda )} \left  \lVert   1_{[0,1)^2+z+h_0} T 1_{L\Lambda^\complement} \right \rVert_p^p \\
= & \sum_{z \in \mathbb Z^2, z+h_0 \in D_{\sqrt 2 }(L \Lambda )} \left  \lVert   1_{[0,1]^2+z+h_0} T 1_{L\Lambda^\complement} \right \rVert_p^p \\ 
\le & \sum_{z \in \mathbb Z^2, z+h_0 \in D_{\sqrt 2 }(L \Lambda )} \left  \lVert   1_{[0,1]^2+z+h_0} T 1_{D_{\operatorname{dist}(z+h_0, L \Lambda^\complement)}^\complement(z+h_0)} \right \rVert_p^p \\
\le & \sum_{z \in \mathbb Z^2, z+h_0 \in D_{\sqrt 2 }(L \Lambda )}  C \frac{\exp\left(- p \lambda \operatorname{dist}(z+h_0, L\Lambda^\complement)^2 \right)}{(1+\lVert z+h_0\rVert)^{pk\varepsilon}} \, .
\end{align}
The last step follows by \autoref{cube norm}. The constant $C$ is independent of $z,h_0$. Now we can integrate this upper bound over $h_0 \in[0,1)^2$. This integral can be resolved by \autoref{grid integral}. Hence, we have
\begin{align}
\left \lVert 1_{L\Lambda} T 1_{L\Lambda^\complement} \right \rVert_p^p 
\le & \int_{[0,1)^2} dh_0 \sum_{z \in \mathbb Z^2, z+h_0 \in D_{\sqrt 2 }(L \Lambda )}  C\frac{ \exp\left(- p \lambda \operatorname{dist}(z+h_0, L\Lambda^\complement)^2 \right)}{(1+\lVert z+h_0 \rVert)^{p k\varepsilon}} \\
= &\int_{D_{\sqrt 2} (L \Lambda )}\frac{  \exp\left(- p \lambda \operatorname{dist}(x, L\Lambda^\complement)^2 \right)} {(1+\lVert x \rVert)^{p k\varepsilon}}dx  \\
=& L^2 \int_{D_{\frac {\sqrt 2}  L }(\Lambda) } \frac{\exp\left(- p \lambda L^2 \operatorname{dist}(x', \Lambda^\complement)^2 \right)} {(1+L\lVert x' \rVert)^{p k\varepsilon}}dx' \\
\le & CL^2 \left(  \int_{\Lambda }\frac{ \exp\left(- p \lambda L^2 \operatorname{dist}(x', \Lambda^\complement)^2 \right)} {(1+L \lVert x' \rVert)^{pk\varepsilon}}dx' + L^{-pk\varepsilon}\left \lvert D_{\frac {\sqrt 2} L}(\Lambda) \setminus \Lambda \right \rvert  \right)\label{317} .
\end{align}

The constant $C$ does not depend on $L$. 
We are left to show, that the term behind $CL^2$ is bounded by $CL^{-1-p k\varepsilon}$.

As $L\ge1$, by \eqref {outer near boundary}, we have
\begin{align}
\left \lvert D_{\frac {\sqrt 2} L}(\Lambda) \setminus \Lambda \right \rvert \le \frac C  L ,
\end{align}
because we can ignore the $\frac 1  {L^2}$ part.  The constant depends on $\Lambda$ and this is the desired estimate.

To estimate the remaining integral, we first use \autoref{Fubini1} and then once more \autoref{Lipschitz slice} to estimate the integral over the enumerator. Thus,
\begin{align}
& \int_{\Lambda } \exp\left(- p \lambda L^2 \operatorname{dist}(x', \Lambda^\complement)^2 \right)dx' \\
= & \int_{\mathbb R}  {p \lambda L^2 h}    \exp\left(- p \lambda L^2 h^2 \right) \left \lvert  \{ x' \in \Lambda \mid  \operatorname{dist}(x', \Lambda^\complement) \le h \} \right \rvert dh \\
\le& \int_0^\infty  {p \lambda L^2 h}  \exp\left(- p \lambda L^2 h^2 \right)  Ch dh \\
=& \int_0^\infty C \exp \left( - (h')^2 \right) (h')^2   \frac{ dh'}{ L } \\
=&\frac C L \, .
\end{align}
In the second to last step, we used the substitution $(h')^2= p \lambda L^2 h^2$. The constant $C$ depends on $p, \lambda$ and in turn on $p,l,k,m,\gamma,{B_0}$ and the decay of $B_\varepsilon,V_\varepsilon$.

To deal with the denominator in \eqref{317}, we use $0 \in \Lambda$. Hence there is an $r>0$, such that $D_{2r}(0) \subset \Lambda$. For the integral over $\Lambda\cap D_r(0)^\complement$, we can bound the denominator from below by $CL^{pk\varepsilon}$ and use the integral estimate above for the enumerator. For the integral over $D_r(0)$ we estimate the enumerator by $Ce^{-L}$ and the denominator by $1$. This finishes the proof.

\end{proof}
Now, we need to consider the final summand in \autoref{iterated resolvent identity}. For that, we need the following theorem, which will be proven in \autoref{last section}.
\begin{theorem} \label{theorem final summand}
Let $\gamma \in \mathbb N,B_\varepsilon,V_\varepsilon$ be $(\gamma,\varepsilon)$ tame, $\Gamma$ be a (finite-length) path in $\mathbb C \setminus \sigma(H)$, $\nu>0$ and let $1 \ge p> \frac 2 {\gamma +3 }  $. Then there is an $n\in \mathbb N$ and a $C>0$, such that we have the following upper bound for any $x_0 \in \mathbb R^2$ and $L>1$: 
\begin{align}
\left \lVert \int_\Gamma1_{[0,1]^2+x_0} (M_{\mathbb N, \zeta} H_\varepsilon)^n \frac 1  {H- \zeta} (H_\varepsilon M_{ \mathbb N, \zeta})^n 1_{L \Lambda^\complement}d \zeta \right \rVert_p \le C (1+\lVert x_0 \rVert)^{\gamma} L^{-\nu}.
\end{align}
The constant $C$ depends on $ {B_0},\gamma, \varepsilon,B_\varepsilon,V_\varepsilon$, but is independent of $x_0$.
\end{theorem}
By the $p$-triangle inequality, the covering of $L \Lambda$ by translated unit boxes, like in the proof of \autoref{penultimate}, and choosing $\nu$ sufficiently large, we arrive at
\begin{corollary} \label{cor final summand}
Let $\gamma \in \mathbb N $, $B_\varepsilon,V_\varepsilon$ be $(\gamma,\varepsilon)$ tame, let $\Gamma$ be a (finite-length) path in $\mathbb C \setminus \sigma(H)$ and let $1 \ge p> \frac 2 {\gamma +3 }  $. Then there is an $n\in \mathbb N$ and a $C>0$, such that  for any $L>1$ we have
\begin{align}
\left \lVert  \int_\Gamma 1_{L\Lambda}   (M_{\mathbb N, \zeta} H_\varepsilon)^n \frac 1 {H- \zeta}  (H_\varepsilon M_{\mathbb N, \zeta} )^n  1_{L\Lambda^\complement}d \zeta \right \rVert_p^p \le C.
\end{align}
The constant $C$ depends on $\Lambda,{B_0},\gamma,p, \varepsilon, B_\varepsilon,V_\varepsilon$.
\end{corollary}

We can now conclude the 
\begin{proof}[Proof of \autoref{p-norm estimate final}]
We assume that $a,b \not \in \sigma(H)$. 
We begin with a fixed Landau level, meaning we even assume $B_0(2l-1)<a<B_0(2l+1)<b < B_0(2l+3)$ for some $l\in \mathbb N$. We choose $\Gamma$ as a path along the circle through $a,b$ with centre $\frac{a+b}2$. We choose $n \in \mathbb N$, as in \autoref{cor final summand}. Now we use \autoref{iterated resolvent identity}. Hence, for any $\zeta \in \operatorname{im} \Gamma$, we have 
\begin{align}
\frac 1 {H-\zeta} = \sum_{k=0}^{2n-1} (-1)^k\frac 1 {H_0-\zeta} \left( H_\varepsilon  \frac 1 {H_0-\zeta} \right)^k + \left(   \frac 1 {H_0-\zeta}H_\varepsilon \right)^n \frac 1 {H-\zeta} \left( H_\varepsilon  \frac 1 {H_0-\zeta} \right)^n.
\end{align}
The path integral over every summand for $0 \le k \le 2n-1$ can be resolved by \autoref{path integral resolve} and then bounded by \autoref{penultimate}. Hence, we have
\begin{align}
&\left \lVert - \frac 1 {2 \pi i}   \int_\Gamma (-1)^k  1_{L \Lambda}\frac 1 {H_0 - \zeta} \left( H_ \varepsilon \frac 1 {H_0 - \zeta} \right)^k 1_{L \Lambda^\complement}\right \rVert_p^p \\
=& \left \lVert (-1)^k \sum_{m=0}^k  1_{L \Lambda} (T_l H_\varepsilon)^m P_l (H_\varepsilon T_l)^{k-m} 1_{L \Lambda^\complement}\right \rVert_p ^p \\
\le & CL^{1-pk \varepsilon}.
\end{align}
In particular, we realize that $P_l$ is the integral over the summand for $k=0$ and hence this summand is cancelled in \eqref{pnef 2}. \autoref{cor final summand} tells us that the path integral over the final summand is even bounded in the $p$-Schatten norm independently of $L$. Another application of the $p$-triangle inequality finishes the proof for a fixed Landau level.

For every $l \in \mathbb N$, such that $a<{B_0}(2l+1) <b$, we choose a circle path, such that the last one hits $\mathbb R$ at $b$, each two neighbouring paths hit $\mathbb R$ at one common point not in $\sigma(H)$,  the first path hits $\mathbb R$ at $a$ and every circle has a real-valued centre. Then we apply the estimate for a single Landau level and the $p$-triangle inequality. 

If there is no Landau eigenvalue between $a$ and $b$, the associated projections are finite dimensional and will lead to an $O(1)$ term with respect to $L$. This also solves the case, where $a \in \sigma(H)$ or $b \in \sigma(H)$. Thus, it finishes the proof.
\end{proof}

\section{Kernel estimates}\label{penultimate section} 

In this section we establish several properties of the Landau Hamilton operator $H_0$ and the  operators $P_l,M_{I,\zeta}$ and in particular, their integral kernels. At the end of this section, we will also include an important integral bound.

We introduce the Laguerre polynomials and their generating function. For any $l \in \mathbb N$, the Laguerre polynomials $\Ln_l$ is given by
\begin{align}
\Ln_l\colon [0,\infty)\to \mathbb R, \quad t \mapsto \sum_{k=0}^{l} \binom l k \frac{(-1)^k }{k!} t^k.
\end{align}

For any $s \in [0,\infty)$, $-1<t<1$, their generating function is given by
\begin{align}
 \sum_{l\in \mathbb N} t^l \Ln_l (s)= \frac 1 {1-t} \exp\left(\frac{-ts}{1-t} \right).
\end{align}

Let $x,y \in \mathbb R^2$. For $l \in \mathbb N$, we define $p_l$ as the integral  kernel of $P_l$, 
\begin{align}
p_l(x,y)  \coloneqq \frac {B_0} {2 \pi}\exp \left(-\frac{B_0}4 \lVert x -y \rVert^2 + i\frac {B_0} 2 \langle x \mid J y \rangle \right) \Ln_l \left( B_0 \lVert x-y \rVert^2/2 \right). \label{pl kernel}
\end{align}

 Furthermore, for $0<t<1$, we define the operator $Q_t \coloneqq \sum_l t^l P_l$. Its integral kernel is given by
\begin{align}
q_t(x,y) \coloneqq& \sum_l t^l p_l(x,y) \label{q_t exp in t}\\
=&\frac {B_0} {2 \pi(1-t)}\exp \left(-\frac {B_0} 4\lVert x -y \rVert^2 + i\frac {B_0} 2 \langle x \mid J y \rangle - \frac{B_0t}{2-2t}\lVert x-y \rVert^2\right) \\
=&\frac {B_0} {2 \pi(1-t)}\exp \left(-\frac{B_0(1+t)  } {4(1-t)} \lVert x -y \rVert^2 + i\frac {B_0} 2 \langle x \mid J y \rangle \right) .\label{def q_t}
\end{align}
We easily calculate
\begin{align}
\left(-i\nabla_x - \frac {B_0} 2 Jx \right) q_t(x,y) &= \left(\frac{iB_0(1+t)  } {2(1-t)} (x-y)  - \frac {B_0} 2 J(x-y)\right)q_t(x,y) \label{H0^0.5q_t} ,
\end{align}
and
\begin{align}
& \left(-i\nabla_x - \frac {B_0} 2 Jx \right)^{\otimes 2} q_t(x,y) \\
=& \left( \left(\frac{iB_0(1+t)  } {2(1-t)} (x-y)  - \frac {B_0} 2 J(x-y)\right)^{\otimes 2} +\left(\frac{B_0(1+t)  } {2(1-t)}  \begin{pmatrix} 1 &0 \\0 &1 \end{pmatrix} - \frac {B_0} 2 J\right) \right) q_t(x,y). \label{p^2q_t}
\end{align}

\begin{lemma} \label{p_l exp bound}
For any $j \in \mathbb N$, there are $C,a>0$, independent of $l,B_0$, such that for any $x,y \in \mathbb R^2$
\begin{align}
\left \lVert (-i\nabla_x-A_0(x))^{\otimes j}p_l(x,y)\right  \rVert &\le B_0^{1+0.5j} C a^l\exp \left( \frac { -B_0 \lVert x-y \rVert^2} 8\right).
\end{align}
\end{lemma}
The norm on the left-hand side is the 2-norm on $\mathbb C^{2^j}$.

\begin{proof}
Using the explicit formula for the Laguerre polynomials, for any $t \ge 0, j' \in \mathbb N, 0< \delta<1 $, we bound the $j'$th differential as follows:
\begin{align}
\lVert \Ln_l^{(j')} (t)\rVert &\le \sum_{k=0}^{l-j'} 2^l  \frac{t^k}{k!} \\
&= \sum_{k=0}^{l-j'} \left( \frac 2 \delta \right)^l  \frac{(\delta t)^k}{k!} \\
&\le \left( \frac 2 \delta \right)^l \exp ( \delta t ) \label{Laguerre bound}.
\end{align}
Each of the $j$ differential operators have to be resolved with the product rule, where we apply the $-i \nabla_x$ to the polynomial, which is resolved by chain rule, and $-i\nabla_x -A_0(x)$  to the exponential. This will always be the exponential times a polynomial expression in $x-y$, taking values in $\mathbb C^{2^j}$.  This leads to the first bound, with a constant $C$ depending only on $j$, as the dependency on $l$ is encoded entirely in the polynomial $L_l$ and its differentials. Thus, we have
\begin{align}
&\left \lVert (-i\nabla_x-A_0(x))^{\otimes j}p_l(x,y)\right  \rVert  \\
\le & C  \sum_ {j'=0}^j \left \lVert \Ln_l^{(j')} \left( \frac {B_0}2\lVert x-y \rVert^2 \right)\right \rVert \left( 1+ \sqrt{B_0} \lVert x-y \rVert \right)^{j} B_0^{1+0.5j} \exp \left(  -\frac {B_0} 4 \lVert x-y \rVert^2 \right) .
\end{align}

By  setting $t= B_0 \lVert x-y \rVert^2/2$ and $\delta = \frac 1 8$ in \eqref{Laguerre bound}, we can finally estimate
\begin{align}
&\left \lVert (-i\nabla_x-A_0(x))^{\otimes j}p_l(x,y)\right  \rVert  \\
\le & C  B_0^{1+0.5j}  \sum_ {j'=0}^j \left \lVert \Ln_l^{(j')} \left( \frac{B_0}2\lVert x-y \rVert^2 \right)\right \rVert \left( 1+ \sqrt{B_0} \lVert x-y \rVert \right)^{j}\exp \left(  -\frac {B_0} 4 \lVert x-y \rVert^2 \right) \\
\le & C  B_0^{1+0.5j} 16^l  \exp \left(\frac {B_0} {16} \lVert x-y \rVert ^2 \right)    \left( 1+ \sqrt{B_0} \lVert x-y \rVert \right)^{j}\exp \left(  -\frac {B_0} 4 \lVert x-y \rVert^2 \right) \\
\le & C  B_0^{1+0.5j} 16^l  \exp \left(\frac {B_0} {8} \lVert x-y \rVert ^2 \right)  \exp \left(  -\frac {B_0} 4 \lVert x-y \rVert^2 \right) \\
\le & C  B_0^{1+0.5j} 16^l  \exp \left(-\frac {B_0} {8} \lVert x-y \rVert ^2 \right).
\end{align}
In the second to last step, we used that polynomials can be bounded by exponentials. The constant $C$ changed, but still only depends on $j$.
\end{proof}

\begin{lemma}  \label{ResolventKernelnew}
Let $I \subset  \mathbb N$ be cofinite, $  \zeta  \in \mathbb C , \zeta \not \in B_0(2I+1)$ and $l_0 \in \mathbb N$, such that $l_0 \ge \max(I^\complement\cup \{\Re\frac{\zeta+B_0}{2B_0}\})$.  Then we have the identity
\begin{align}
B_0 M_{I,\zeta} =& \int_0^1 t^{-\zeta/B_0} \left(Q_{t^2}-\sum_{ l \le l_0}t^{2l} P_l\right) dt+  \sum_{l\in I, l \le l_0} \frac{ P_l}{(2l+1-\zeta/B_0)} .
\end{align}

\end{lemma}
\begin{proof}
The idea of this proof is the formal identity
\begin{align}
\int_0^1 \sum_{l \in I} t^{2l-\zeta/B_0} P_l dt &= \sum_{l\in I} \frac 1 {1+2l-\zeta/B_0} P_l .
\end{align}

Now we need to establish the precise meaning of this identity. First, we note that $t^{-\zeta/B_0}=\exp(-\zeta/B_0\ln(t))$ is well defined, as $t >0$. If $\Re(\zeta)/B_0 \ge 2l+1$, then the integral of the summands for $l$ will not exist, which is the reason we introduced $l_0$. We bounded the real part of $\zeta$ a little stronger than necessary to make the proof easier.
Hence, we have
\begin{align}
\int_0^1 \sum_{l>l_0} t^{2l-\zeta/B_0} P_l dt = \sum_{l>l_0} \frac 1 {1+2l-\zeta/B_0} P_l.
\end{align}

For any single $l> l_0$, the integral exists as a Bochner integral with respect to the operator norm.  \autoref{orth op fam} finishes the proof.
\end{proof}

We will deal with a few integral kernels that have a singularity at the diagonal. To describe such a singularity, for any $s\in \mathbb R$, we introduce 
\begin{align}
b_s \colon \mathbb R^2 \to [0,\infty), \quad (x,y) \mapsto \begin{cases}  -1_{D_{\frac {1 }{\sqrt {B_0}}}(0)}(x-y) \ln (\sqrt {B_0}\lVert x-y\rVert)  & s=0, \\
   1_{D_{\frac {1} {\sqrt {B_0}}}(0)}(x-y)\frac {1} {\lVert x-y\rVert^s} & s \not =0 .\end{cases} \label{def bs}
\end{align}

\begin{lemma} \label{1.resolvent bound}
Let $I \subset  \mathbb N$ be cofinite. Then there is a function $F \in \Lp^\infty_{loc}(\mathbb C\setminus(2I+1))$, 
 such that the following pointwise upper bounds hold for all $x,y \in \mathbb R^2, x \not=y$ and $\zeta \in \mathbb C \setminus{B_0}(2I+1)$:
\begin{align}
\lvert \operatorname{iker} M_{I,\zeta}(x,y)\rvert \le& F\left(\frac \zeta {B_0}\right)\left( b_0(x,y)+  \exp\left( -\frac {B_0} 8 \lVert x-y \rVert^2 \right)\right),\\
\lVert \operatorname{iker} (-i\nabla -A_0)  M_{I,\zeta}(x,y)\rVert \le&   F\left(\frac \zeta {B_0}\right)\left( b_1(x,y) + \sqrt {B_0}\exp\left( -\frac {B_0} 8 \lVert x-y \rVert^2   \right)\right),\\ 
\lVert  (-i\nabla_x -A_0(x))^{\otimes 2} \operatorname{iker} M_{I,\zeta}(x,y)\rVert \le&  F\left(\frac \zeta {B_0}\right)\left( b_2(x,y)+  {B_0}\exp\left( -\frac {B_0} 8 \lVert x-y \rVert^2 \right)\right).
\end{align}
\end{lemma}
\begin{remark}
The last inequality is structurally different, because the implied operator $(-i \nabla -A_0)^{\otimes 2} M_{I,\zeta}$ does not have a \emph{nice} integral kernel. The differential of the integral kernel can still be considered but is not $\Lp^1$ with respect to $y$ for any fixed $x$ and hence not a \emph{nice} integral kernel. In general, this kernel does not fully describe the operator. 
\end{remark}

\begin{proof}The set $I\subset \mathbb N$ is fixed throughout the proof.

For any $t \in [0,1), l \in \mathbb N, j \in \{0,1,2\}$, we define
\begin{align}
 q_{t,j}(x,y) &\coloneqq (-i\nabla_x-A_0(x))^{\otimes j}q_t(x,y), \\
p_{l,j} (x,y)&\coloneqq (-i\nabla_x-A_0(x))^{\otimes j}p_l(x,y).
\end{align}
As $q_{t,j}, p_{l,j}$ are \emph{nice} integral kernels, we can apply dominated convergence and see that
\begin{align}
 q_{t,j}(x,y) &=\operatorname{iker}\left((-i\nabla-A_0)^{\otimes j}Q_t\right)(x,y), \\
p_{l,j} (x,y)&=\operatorname{iker} \left( (-i\nabla-A_0)^{\otimes j}P_l\right)(x,y).
\end{align}
We choose $l_0 \in \mathbb N$ minimal, such that $(2l_0-1){B_0} > \Re \zeta$ and $l_0 \ge \max(I^\complement)$. Now we  use the representation established in \autoref{ResolventKernelnew}. To prove, that for $j \in \{0,1\}$, the operators have integral kernels, we want to use \autoref{Bochner kernel}. Hence, we only need to show, that the following inequality holds, in order to finish the proof for $j=0,1$:
\begin{align}
& \int_0^1 \left \lVert t^{-\zeta/{B_0}} \left(q_{t^2,j}(x,y)-\sum_{ l \le l_0} t^{2l} p_{l,j}(x,y)\right) \right \rVert dt+  \sum_{l\in I, l \le l_0}\frac{ \left \lVert p_{l,j}(x,y) \right \rVert}{(2l+1-\zeta/{B_0})} \\
\le & B_0 F\left(\frac \zeta {B_0}\right)\left( b_j(x,y)+  {B_0}\exp\left( -\frac {B_0} 8 \lVert x-y \rVert^2 \right)\right) .
\end{align}
For $j=2$, however, we need to consider, that as the integrand is smooth on $(0,1)$ and the summands at the end are smooth, we can try to exchange the integral with the differential operator $(-i \nabla-A_0)$. This will work, if the absolute value of the differential is integrable, by dominated convergence. Hence above integral bound also covers the case $j=2$ and we will now proceed to bound all terms at the same time by choosing $j\in \{0,1,2\}$. 
We want to use \autoref{p_l exp bound} to bound the first integral on the interval $(0, t_0 )$ and the sums. Hence,
\begin{align}
&  \int_0^{t_0} \left \lVert t^{-\zeta/{B_0}} \left(q_{t^2,j}(x,y)-\sum_{ l \le l_0} t^{2l} p_{l,j}(x,y)\right) \right \rVert dt \\
\le &\int_0^{t_0} \sum_{l>l_0}t^{2l-\Re(\zeta)/{B_0}} \lvert p_{l,j}(x,y) \rvert dt \\
\le &\int_0^{t_0} \sum_{l>l_0}t^{2l-\Re(\zeta)/{B_0}} C{B_0}^{1+0.5j}a^l \exp\left(- \frac {B_0} 8 \lVert x-y \rVert^2 \right)  dt \\
= &C{B_0}^{1+0.5j}  \sum_{l>l_0} \frac{ ( t_0^2 a)^l t_0}{(2l+1-\Re( \zeta)/{B_0})t_0^{\Re (\zeta)/{B_0}} } \exp\left(- \frac {B_0} 8 \lVert x-y \rVert^2 \right) \\
\le & F(\zeta/{B_0}){B_0}^{1+0.5j}\exp\left(- \frac {B_0} 8 \lVert x-y \rVert^2 \right) .
\end{align}
The last step holds, if $t_0^2a<1$, so we fix such a $t_0$ now\footnote{Actually $a=16$, so we could choose for example $t_0=0.1$, but the value is not relevant.}. The function $F_0$ is in $\Lp^\infty_{loc}(\mathbb C \setminus B_0(2I+1) )$, as $l_0$ is chosen locally bounded in $\zeta/{B_0}$. For fixed $l_0$ the function $F_0$ is continuous.
The next step is bounding the remaining finite sum terms. Here, we will use, that $l \le l_0$ and hence $a^l \le C$. Thus,
\begin{align}
&  \int_{t_0}^1 \left \lVert \sum_{ l \le l_0} t^{2l-\zeta/{B_0}} p_{l,j}(x,y) \right \rVert dt + \sum_{l\in I, l \le l_0}  \frac{ \left \lVert p_{l,j}(x,y) \right \rVert}{(2l+1-\zeta/{B_0})}   \\
\le&C{B_0}^{1+0.5j}\left( \sum_{ l \le l_0} \left( \int_{t_0}^1 t^{2l-\Re(\zeta)/{B_0}} dt  \right) +  \sum_{l\in I, l \le l_0} \left( \frac 1 {\lVert 2l+1-\zeta/{B_0} \rVert }  \right)\right) \exp\left(- \frac {B_0} 8 \lVert x-y \rVert^2 \right) \\ 
\le & F(\zeta/{B_0}) {B_0}^{1+0.5j} \exp\left(- \frac {B_0} 8 \lVert x-y \rVert^2 \right) .
\end{align}
The function $F_1$ is in $\Lp^\infty_{loc}(\mathbb C \setminus B_0(2I+1))$ by the same argumentation as $F_0$.
We will now turn our attention to the last remaining term. It is given by
\begin{align}
\int_{t_0}^1 \left \lVert q_{t^2,j}(x,y) t^{-\zeta/B_0} \right \rVert dt.
\end{align}
The integrand is given by \eqref{H0^0.5q_t} for $j=1$ and by \eqref{p^2q_t} for $j=2$. Only in the following lines, we  denote by $j \mapsto \delta(j,2)$ the function, that is $1$, if $j=2$ and $0$ otherwise.  We introduce the parameter $h \coloneqq \sqrt {B_0}  \lVert x-y \rVert$ and estimate 
\begin{align}
&\int_{t_0}^1 \left \lVert t^{-\zeta/B_0}q_{t^2,j}(x,y) \right \rVert dt  \\
\le &\left(t_0^{-\Re(\zeta)/{B_0}}+1\right) \int_{t_0}^1\frac{C{B_0}}{1-t^2} \left( \left( \frac{ \sqrt {B_0} h} {2(1-t^2)}\right)^j +\frac{ \delta(j,2) {B_0} } {2(1-t^2)}   \right)\exp \left( -\frac {1+t^2}{4(1-t^2)} h^2 \right)dt \\
\le  & \int_{0}^1\frac{F(\zeta/{B_0}) {B_0}}{1-t}\left( \left( \frac{ \sqrt {B_0} h} {2(1-t)}\right)^j + \frac{ \delta(j,2) {B_0} }{2(1-t)} \right)\exp \left( -\frac {1}{5(1-t)} h^2 \right)  dt.
\end{align}
In the last step, we used the fact, that $\frac{1+t^2}{1+t} \ge 2 \sqrt 2 -2> \frac 4 5$ to bound the factor in the exponential. The function $F_2$ is just continuous on $\mathbb C$.

We want to do a change of variables to $s \coloneqq \frac{ h^2} {5(1-t)}$.  The interval is changed to $(h^2/5, \infty)$ and the determinant is $h^2/(5s^2)$. In total we have
\begin{align}
&\int_{t_0}^1\left \lVert t^{-\zeta/B_0}q_{t^2,j}(x,y)  \right \rVert dt \\
\le & F\left(\frac \zeta {B_0} \right){B_0}^{1+0.5j}\int_{h^2/5}^\infty \frac{s}{h^2}\left( \left( \frac{ s} {h}\right)^j + \frac{ \delta(j,2) s }{h^2} \right)\exp \left( -s \right)   \frac{ h^2}{s^2}ds \\
\le & F\left(\frac \zeta {B_0} \right){B_0}^{1+0.5j}\int_{h^2/5}^\infty \frac{1}{s}\left( \left( \frac{ s} {h}\right)^j + \frac{ \delta(j,2) s }{h^2} \right)\exp \left( -s \right)   ds  \eqqcolon \Theta .
\end{align}
If $h>1$, we can bound the integrand by $C\exp(-\frac 5 8 s)$. The reduction in the exponent takes care of the factor $s$, that appears in the case $j=2$. Negative powers of $h$ can be bounded by one. The integral can then be resolved and we have
\begin{align}
\Theta \le &C  F\left(\frac \zeta {B_0} \right){B_0}^{1+0.5j} \exp \left( -\frac 5 8 \frac {h^2} 5\right) \\
=& C F\left(\frac \zeta {B_0} \right){B_0}^{1+0.5j} \exp \left( -\frac {B_0} 8 \lVert x- y \rVert^2\right) .
\end{align}
This is the desired upper bound. 

 If $h\le 1$, $j >0$, we can set the lower interval limit to $0$ and get an integrable function in $s$ multiplied by $h^{-j}$. This gives us
\begin{align}
\Theta \le  C F\left(\frac \zeta {B_0} \right) {B_0}^{1+0.5j}  h^{-j} \le & C F_2\left(\frac \zeta {B_0} \right) {B_0} b_j(x,y),
\end{align}
which is the desired upper bound.

Finally, if $h \le 1$, $j=0$, we get a constant from the integral starting at $\frac 1 5$. For the integral up to $\frac 1 5 $, we can bound the integrand by $\frac 1 s$. Hence, the remaining integral is bounded by  $C(1-\ln (h^2 ) )\le C (1+ b_0(x,y))$. Once again, this is the desired result.
\end{proof}

We need one very important bound, which will have multiple uses later.

\begin{lemma} \label{2.resolvent bound}
Let $u_1,u_2,u_3 \colon \mathbb R^2 \to \mathbb R^+$ be functions, such that $\ln \circ u_j$  is Lipschitz with Lipschitz constant $C_{lip}>0$. Let $0\le s_1,s_2 < 2$ and $\lambda>0$ be real numbers. Then there is a constant $C>0$, depending only on ${B_0},s_1,s_2,\lambda$ and $C_{lip}$, such that for all $x,y \in \mathbb R^2, x \not=y$  we have the estimate
\begin{align}
&\int_{\mathbb R^2}\frac{ \left( b_{s_1} (x,y) + \exp (-{B_0}\lambda \lVert x-y \rVert^2 ) \right) \left( b_{s_2} (y,z) + \exp (-{B_0}\lambda \lVert y-z \rVert^2 ) \right) }{u_1(x)u_2(y)u_3(z)}dy \\
\le & \frac{C b_{s_1+s_2-2}(x,z) + C \exp \left( -\frac{{B_0} \lambda}3 \lVert x-z \rVert^2 \right)}{u_1(x)u_2(x)u_3(x)} \, . \label{eq:2rb}
\end{align}
If $1/(u_1u_2u_3) \in \Lp^2(\mathbb R^2)$ and $s_1+s_2 <3$, then the integral kernel is Hilbert--Schmidt. 
\end{lemma}
This is to be used together with \autoref{1.resolvent bound} with $\lambda=\frac 1 8$. The general $\lambda$ is included to be able to chain more resolvents inductively. 
 As all summands in the integral are positive, we may assume that they have the same constants in front.
\begin{proof}
We first need two minor results. Let $a,b \in \mathbb R^2, j \in\{1,2,3\}$. Then for any $\delta>0$, we have
\begin{align}
\frac{ u_j(a)}{u_j(b)} =& \exp \left( \ln \circ u_j(a)- \ln  \circ u_j(b) \right) \\
\le & \exp \left( C_{lip} \lVert a-b \rVert \right) \label{loglip1}\\
\le& \exp \left( \delta \lVert a-b \rVert ^2 +\frac{C_{lip}^2 }{4 \delta} \right) \\
=& C(C_{lip},\delta) \exp\left( \delta \lVert a-b \rVert^2 \right) \label{minor1}.
\end{align}
We used the Young inequality. Furthermore (for any $x,y,z \in \mathbb R^2$) we have the identity
\begin{align}
\lVert x-y \rVert^2 +\lVert y-z \rVert^2 = \frac 12 \lVert x-z \rVert^2 + 2 \left \lVert y- \frac{x+z}2 \right \rVert^2.
\end{align}
We write $R\coloneqq \frac 1 {\sqrt {B_0}}$. Let us begin with the left-hand side of \eqref{eq:2rb} and just write out most of the Hölder estimates. Hence,
\begin{align}
LHS \le & \frac C {u_1(x)u_3(z)} \\
& \Bigg(\int_{D_{R}(x)} b_{s_1}(x,y)b_{s_2}(y,z) dy \left \lVert \frac 1 {u_2(\blank)} \right \rVert_{\Lp^\infty  \left (D_{R} (x) \right)}\\
&\quad+ \exp \left( - \frac {{B_0}\lambda} {2} \lVert x-z \rVert^2 \right) \int_{\mathbb R^2}\frac{ \exp\left(-2{B_0}\lambda\left \lVert y- \frac{x+z}2 \right \rVert^2 \right) }{u_2(y)}  dy \\
& \quad + \left \lVert b_{s_1}(x, \blank) \right\rVert_{\Lp^1 \left (D_{R} (x) \right)}  \left \lVert \exp \left( - {B_0}\lambda \lVert \blank-z \rVert^2 \right)    \right\rVert_{\Lp^\infty \left (D_{R} (x) \right)} \left \lVert \frac 1 {u_2(\blank)} \right \rVert_{\Lp^\infty  \left (D_{R} (x) \right)} \\
& \quad  +\left \lVert  b_{s_2}( \blank,z )   \right\rVert_{\Lp^1 \left (D_{R} (z) \right)}  \left \lVert \exp \left( - {B_0}\lambda \lVert x-\blank \rVert^2 \right)  \right\rVert_{\Lp^\infty \left (D_{R} (z) \right)}  \left \lVert \frac 1 {u_2(\blank)} \right \rVert_{\Lp^\infty  \left (D_{R} (z) \right)}\Bigg) .
\end{align}

The $\Lp^\infty$ norms of the non-exponential terms can be bounded by a constant times the function evaluated at the centre, where the constant is given by \eqref{loglip1}, using $a$ as the centre of the ball and $b$ as any point in the ball. For the $\Lp^\infty$ norms of the exponential terms, we use \autoref{gauss disc} with $x_0 \coloneqq y-z$. We are left to estimate the four $\Lp^1$ norms, some of which are written as integrals. The last two $\Lp^1$ norms can be bounded by a constant and that is sufficient. For the exponential integral, we first use \eqref{minor1} with $\delta=B_0$ to replace the $u_2(y)$ in the denominator by $u_2((x+z)/2)$, getting a different Gaussian in the numerator, and then we can just bound its integral. With all of these, we get
\begin{align}
LHS \le & \frac C {u_1(x)u_3(z)} \times \\
& \Bigg(\int_{D_{R}(x)} b_{s_1}(x,y)b_{s_2}(y,z) dy \frac 1 {u_2(x)}+ \frac{ \exp \left( -  {{B_0}\lambda}/2\lVert x-z \rVert^2 \right) } {u_2( (x+z)/2)} \\
& \quad +  \exp \left( - \frac{{B_0}\lambda} 2 \lVert x-z \rVert^2 \right)      \frac 1 {u_2(x)}   +  \exp \left( - \frac{{B_0}\lambda} 2 \lVert x-z \rVert^2 \right)      \frac 1 {u_2(z)}  \Bigg) .
\end{align}
If we apply \eqref{minor1} again, we can get the desired bound for the last three summands. So, we only need to get the same bound for the first summand. If $\lVert x-z \rVert >2R$, the first summand vanishes. Otherwise, the term $1/u_3(z)$ can be bounded by $C/u_3(x)$ by \eqref{minor1}.  In the case $2R \ge \lVert x-z \rVert \ge  R/2$, we just bound the integral by a constant depending on $R$, which can then be bounded by a constant times the Gaussian.  We are left to consider the case $\lVert x-z \rVert < R/2$. So, we are left to bound the integral
\begin{align}
\int_{D_{R}(x)} b_{s_1}(x,y)b_{s_2}(y,z) dy.
\end{align}
We have $b_{s}(x, \blank) \in \Lp^p $ for any $1 \le p < 2/s$ and $b_s$ is symmetric in $x,y$. Hence, if $s_1+s_2<2$, we can bound this by a constant (independent of $x,z$) using Hölder. This can then by bounded by the Gaussian, as $\lVert x-z \rVert \le 2R$. We are left with the case $s_1+s_2 \ge 2$, where we want to bound the integral by $b_{s_1+s_2-2}(x,z)+ C$. As $s_1,s_2 <2$, we have $s_1,s_2>0$. Let $e_1 \in \mathbb R^2$ be the standard unit vector and let $D_{r_1,r_2}(0)$ be the annulus between the two radii $r_1 \le r_2$. Then we have
\begin{align}
&\int_{D_{R}(x)} b_{s_1}(x,y)b_{s_2}(y,z) dy\\
\le & \int_{D_{R}(x)} \frac 1 {\lVert x-y \rVert^{s_1} \lVert y-z \rVert^{s_2}}dy\\
= & \int_{D_{R}(0)} \frac 1 {\lVert y \rVert^{s_1} \lVert y-(z-x) \rVert^{s_2}}dy\\
= & \int_{D_{R\lVert x-z \rVert^{-1}}(0)} \frac {\lVert x-z \rVert ^{2-s_1-s_2}} {\lVert y \rVert^{s_1} \lVert y- e_1 \rVert^{s_2}}dy\\
\le &\frac {\int_{D_2(0)} \lVert y \rVert^{-s_1} \lVert y- e_1 \rVert^{-s_2}dy} {\lVert x-z \rVert ^{s_1+s_2-2}} + \int_{D_{2, R\lVert x-z \rVert^{-1}}(0)} \frac {\lVert x-z \rVert ^{2-s_1-s_2}}  {\lVert y \rVert ^{s_1}\lVert y-e_1 \rVert^{s_2}}dy\\
\le &\frac C {\lVert x-z \rVert ^{s_1+s_2-2}}+ \int_{D_{2, R\lVert x-z \rVert^{-1}}(0)} \frac {C\lVert x-z \rVert ^{2-s_1-s_2}}  {\lVert y \rVert^{s_1+s_2}}dy\\
\le & C b_{s_1+s_2-2}(x,z).
\end{align}
In the final step, we have to consider the case $s_1+s_2=2$ separately. In this case, the integral at the end yields  the term $b_0(x,z)$ up to a constant. In the case $s_1+s_2>2$, the integral over $\lVert y \rVert^{-s_1-s_2}$ can be bounded by a constant, independent of $x,z$ and we are left with the correct singularity at the diagonal. This finishes the proof of the upper bound. 

If $1/(u_1u_2u_3) \in \Lp^2$ and $s_1+s_2<3$, we get
\begin{align}
&C\int_{\mathbb R^2}dx \int_{\mathbb R^2} dz \left( \frac{b_{s_1+s_2-2}(x,z) + \exp \left( -\frac{{B_0} \lambda}3 \lVert x-z \rVert^2 \right)}{u_1(x)u_2(x)u_3(x)} \right)^2 \\
=&C\int_{\mathbb R^2}dx \int_{\mathbb R^2} d(x-z)\left( \frac{b_{s_1+s_2-2}(x,z) + \exp \left( -\frac{{B_0} \lambda}3 \lVert x-z \rVert^2 \right)}{u_1(x)u_2(x)u_3(x)} \right)^2 \\
\le&\int_{\mathbb R^2}dx\left( \frac{C}{\left(u_1(x)u_2(x)u_3(x)\right)^2} \right) \le C.
\end{align}
Hence the integral kernel is Hilbert--Schmidt.
\end{proof}

\begin{corollary} \label{iterated comp kernel}
Let $n \in \mathbb N$ and for any $0 \le i \le n$, let there be an operator $K_i$ with integral kernel $k_i$ on $\Lp^2(\mathbb R^2)$,  log-Lipschitz functions $u_i,v_i\colon \mathbb R^2 \to \mathbb R^+$, $\lambda_i >0$, and  $0 \le s_i <2$. Assume the integral kernels $k_i$ satisfy the upper bound
\begin{align}
u_i(x) \lvert k_i(x,y)\rvert v_i(y) \le C b_{s_i} (x,y) +C \exp \left( - \lambda_i \lVert x-y \rVert^2 \right), \label{icp assump}
\end{align}
for any $x \not =y$. Define $K \coloneqq \prod_{i=0}^n K_i$ and let 
\begin{align}
s = -2n+ \sum_{i =0}^n s_i .
\end{align}
Then $K$ has an integral kernel $k$ and there are $\lambda>0, C>0$, such that for any $x \not =y$, we have the inequalities
\begin{align}
\lvert k(x,y) \rvert \le \frac{C b_s(x,y) +C \exp \left( -\lambda \lVert x-y \rVert^2 \right)}{\prod_{i=0}^n u_i(x)v_i(x) }\label{icp eq1}, \\
\lvert k(x,y) \rvert \le \frac{C b_s(x,y) +C \exp \left( -\lambda \lVert x-y \rVert^2 \right)}{\prod_{i=0}^n u_i(y)v_i(y) }. \label{icp eq2}
\end{align}

\end{corollary}

For $s<0$, we can replace $b_s$ by $0$ in \eqref{icp eq1} and \eqref{icp eq2}, as $b_s$ is bounded and can be absorbed in the Gaussian.
\begin{proof}
The case $n=0$ follows by \eqref{minor1}. We continue with the case $n=1$. By \autoref{2.resolvent bound}, we only have to show that $K_0K_1$ has is an integral operator and that for any $x,z \in \mathbb R^2$ with $x \neq z$, we have
\begin{align}
\operatorname{iker}K_0K_1(x,z)= \int_{\mathbb R^2} dy \operatorname{iker}K_0 (x,y) \operatorname{iker}K_1 (y,z) . \label{K0K1 has kernel}
\end{align}
To do so, it is sufficient to find a function space $\mathcal Y \supset C^0_c(\mathbb R^2)$, on which $K_0$ and $K_1$ are continuous.  We claim the topological vector space
\begin{align}
\mathcal Y \coloneqq \bigcap_{\lambda \in \mathbb R} \left \{ f({\cdot}) \exp( \lambda \lVert {\cdot}\rVert)  \in \Lp^\infty(\mathbb R^2)\right\} 
\end{align}
does the trick. 

We observe that any log-Lipschitz function $u \colon \mathbb R^2 \to \mathbb R^+$ with log-Lipschitz constant $C_{Lip}$ satisfies for any $x \in \mathbb R^2$ that
\begin{align}
 u(0) \exp(-C_{Lip} \lVert x \rVert ) \le u(x) \le  u(0) \exp(C_{Lip} \lVert x \rVert ) .
\end{align}
Hence, such a function defines a continuous multiplication operator on $\mathcal Y$. By the assumption \eqref{icp assump}, the operators $K_i$ can each be written as a product of two such multiplication operators and a nice integral operator $K_i'$ satisfying the kernel estimate
\begin{align}
\operatorname{iker} K_i' (x,y) \le b_{s_i}(x,y)+ \exp\left(- \lambda_i \lVert x-y \rVert^2 \right)
\end{align}
for any $x,y \in \mathbb R^2$ with $x \neq y$. We observe that such an integral operator is bounded on $\mathcal Y$. Now, by Fubini we can conclude \eqref{K0K1 has kernel}. This finishes the case $n=1$ with $\lambda= \frac 1 3\min\{\lambda_1,\lambda_2\}$.

As the resulting estimate for $K_0K_1$ is of the same form as the required estimate in \eqref{icp assump}, the induction over $n$ follows trivially.
\end{proof}

\section{Proof of \autoref{cube norm} and \autoref{theorem final summand}} \label{last section}

We will first briefly summarize the approach for both proofs. We will start by conjugating with the unitary operator $U_{x_0}$, as defined in \autoref{U_{x_0}}. Then, we can show that the operators we produce this way are Hilbert--Schmidt operators from $\Lp^2(\mathbb R^2)$ to $H^{\gamma+2}([0,1]^2)$ using the quasi isometry $D_{\gamma+2}$, that we have constructed in \autoref{H quasi iso} and  some commutator relations to move the differentials around. The proofs will conclude with  \autoref{Sob emb Schatten2}.

 We denote by $t_l$ the integral kernel of $T_l$. By \autoref{1.resolvent bound}, we can only apply one full differential in $x$ or $y$ to $t_l$, before we get a function, that is not a \emph{nice} integral kernel anymore. However, the operator $P_l$ has a smooth integral kernel, which is why we would like to move differentials over to it. We will see that we can apply two differentials after $M_{I, \zeta}$ and still remain with a bounded operators, that is (in general) not an integral operator in \autoref{2nd diff resolv}. 

We also still need to prove \autoref{4-Schatten}. However, it is convenient to prove a more general integral kernel bound along with it. For that, we need to introduce some new notation.

For any $x_0 \in \mathbb R^2$, $j \in \mathbb N$, we introduce the three multiplication operators, which are given for any $x \in \mathbb R^2$ by
\begin{align}
B_{\varepsilon,x_0}^{(j)}(x) &\coloneqq (-i)^j \left( \left(\prod_{h=1}^j  \partial_{\theta_h}\right) B_\varepsilon(x+x_0)\right)_{\theta \in \{1,2\}^j}, \label{def B_eps^j} \\
A_{\varepsilon,x_0}^{(j)}(x) &\coloneqq B_\varepsilon^{(j)}(\blank+x_0) * \frac{ J\blank}{2\pi \lVert \blank \rVert^2}  (x),\\
W^{(j)}_{\varepsilon,x_0}(x) &\coloneqq  (-i)^j \left( \left(\prod_{h=1}^j  \partial_{\theta_h}\right) W_\varepsilon(x+x_0)\right)_{\theta \in \{1,2\}^j}, \\
H_{\varepsilon,x_0}^{(j)} &\coloneqq     A^{(j)}_{\varepsilon,x_0}\cdot \left(-i \nabla - A_0 \right)+W^{(j)}_{\varepsilon,x_0} \label{def H_{eps,x0}} .
\end{align}
The last equation defines a non-multiplication operator.  
[We have not defined $A_{x_0}$ and $H_{x_0}$, as this may lead to confusion with $A_0$ and $H_0$, if we set $x_0=0$.] The scalar product in the definition of $H^{(j)}_{\varepsilon,x_0}$ reduces the final component of $A^{(j)}_{\varepsilon,x_0}$, which originates from the convolution with the $\mathbb R^2$ valued function $\frac{ J {\cdot}}{2 \pi \lVert {\cdot \rVert^2}}$. We will write $X_{\varepsilon,x_0}$ for $X^{(0)}_{\varepsilon,x_0}$ for $X \in \{A,W,B\}$.  

We observe that by \autoref{dA def} for $f= B_\varepsilon^{(j)}$, we have
\begin{align}
 A^{(j)}_{\varepsilon,x_0}\cdot \left(-i \nabla - A_0 \right) = \left(-i \nabla - A_0 \right) \cdot  A^{(j)}_{\varepsilon,x_0},
\end{align}
where the scalar product on both sides reduces the  final component of $A^{(j)}_{\varepsilon,x_0}$, which originates from the convolution with the $\mathbb R^2$ valued function $\frac{ J {\cdot}}{2 \pi \lVert {\cdot \rVert^2}}$, as above.
Hence, we have, with the same scalar product,
\begin{align}
H_{\varepsilon,x_0}^{(j)} &=      \left(-i \nabla - A_0 \right) \cdot A^{(j)}_{\varepsilon,x_0}+W^{(j)}_{\varepsilon,x_0} \label{def H_{eps,x0} alt} .
\end{align}
The idea behind these definitions is, as we hinted at in the introduction to this section, that by conjugating with the unitary operator $U_{x_0}\colon \Lp^2(\mathbb R^2 ) \to \Lp^2(\mathbb R^2)$, as defined in \autoref{U_{x_0}}, we observe the identity
\begin{align}
\left \lVert 1_{[0,1]^2-x_0} (T_l H_\varepsilon)^m P_l (H_\varepsilon T_l)^{k-m} 1_{D^\complement_R(x_0)}\right \rVert_p
=& \left \lVert 1_{[0,1]^2} (T_l H_{\varepsilon,x_0})^m P_l (H_{\varepsilon,x_0} T_l)^{k-m} 1_{D^\complement_R(0)}\right \rVert_p, \label{conjugate U_x0}
\end{align}
as the $p$-Schatten norm is unitarily invariant. Something similar applies for the proof of \autoref{theorem final summand}.

It is now time to prove \autoref{4-Schatten}. However, as we will need a more general statement, we will prove that instead.

\begin{lemma} \label{H eps M kernel}
Let $\gamma \in \mathbb N$, $V_\varepsilon, B_ \varepsilon$ be $(\gamma,\varepsilon)$ tame, $I \subset \mathbb N$ cofinite, and let $\zeta \in \mathbb C \setminus (2I+1) {B_0}$. Furthermore, let $x_0 \in \mathbb R^2$ and $\mathbb N \ni d \le \gamma$. Then there is a function $F \in L^\infty_{loc}(\mathbb C \setminus (2I+1) {B_0})$ and a real number $\lambda>0$, such that for any $x,y \in \mathbb R^2$ with $x \neq y$, we have the upper bound
\begin{align}
\left \lVert \operatorname{iker}H_{\varepsilon,x_0}^{(d)} M_{I, \zeta} (x,y) \right \rVert \le F( \zeta) \frac{ b_1(x,y) + \exp(-\lambda \lVert x-y \rVert^2) }{(1+ \lVert x+ x_0 \rVert)^\varepsilon }.
\end{align}
In particular, this is a nice integral kernel and the operator norm of $H_{\varepsilon,x_0}^{(d)} M_{I, \zeta}$ is bounded independently of $x_0$. Furthermore, we have the estimate
\begin{align}
\left \lVert H_{\varepsilon} M_{I, \zeta} \right \rVert_{4n_0} \le F(\zeta),
\end{align}
where $n_0$ is the smallest integer such that $2n_0 \varepsilon >1$.
\end{lemma}
This lemma generalizes \autoref{4-Schatten}. The operators $H_{\varepsilon,x_0}^{(d)}$ and $M_{I, \zeta}$ have been defined in \eqref{def H_{eps,x0}} and \eqref{def M_Iz}. The function $b_s$ has been defined in \eqref{def bs}.

\begin{proof}
We can estimate pointwise for $x \not = y$, using the assumption that $B_\varepsilon,V_\varepsilon$ are $(\gamma,\varepsilon)$ tame, \autoref{dA def}, and \autoref{1.resolvent bound}. Thus, we have
\begin{align}
&\left \lVert \operatorname{iker} A_{\varepsilon,x_0}^{(d)} (-i\nabla -A_0) M_{I,\zeta}  (x,y) \right \rVert \\ 
\le & \frac {C}{(1+\lVert x+x_0 \rVert )^{\varepsilon}}\left  \lVert  \operatorname{iker}( - i \nabla  - A_0  ) M_{I,\zeta} (x,y) \right \rVert\\
\le &  F\left(\frac \zeta{B_0}\right)\frac {b_1(x,y)+\sqrt {B_0} \exp \left( - \frac {B_0} 8 \lVert x-y \rVert^2 \right)}{(1+\lVert x+x_0 \rVert )^{\varepsilon}}.
\end{align}
And now for the other part, we observe that $W_\varepsilon = A_\varepsilon^2 + V_\varepsilon \in W^{\gamma,\infty}_{(\varepsilon)}(\mathbb R^2, \mathbb R)$ and can then use \autoref{1.resolvent bound} to see
\begin{align}
&\left \lVert \operatorname{iker} W_{\varepsilon,x_0}^{(d)} M_{I,\zeta}   (x,y) \right \rVert \\ 
\le &  \frac {C}{(1+\lVert x+x_0 \rVert )^{\varepsilon}} \lVert \operatorname{iker} M_{ I, \zeta} (x,y) \rVert \\
\le &  F\left(\frac \zeta{B_0}\right)\frac {b_0(x,y)+ \exp \left( - \frac {B_0} 8 \lVert x-y \rVert^2 \right)}{(1+\lVert x +x_0\rVert )^{\varepsilon}} \\
\le &  F\left(\frac \zeta{B_0}\right)\frac {b_1(x,y)+\sqrt {B_0} \exp \left( - \frac {B_0} 8 \lVert x-y \rVert^2 \right)}{(1+\lVert x+x_0 \rVert )^{\varepsilon}}.
\end{align}
In the last step we used $b_0 \le Cb_1$ and $1 = C \sqrt {B_0}$. This shows the first claim.

We use properties we denoted as powers and Hilbert--Schmidt kernel of the $p$-Schatten norms. Hence the $4n_0$-Schatten norm of $T$ can be calculated as the $4n_0$th root of the square integral of the integral kernel of $(TT^*)^{n_0}$. We note, that $u(x) \coloneqq (1+\lVert x \rVert)^\varepsilon$ is log-Lipschitz. We want to use \autoref{iterated comp kernel}. Hence, we define for $0 \le i \le 2n_0-1$
\begin{align}
K_i \coloneqq \begin{cases}  H_{\varepsilon} M_{I,\zeta} & i \text{ even}, \\
\left( H_\varepsilon M_{I,\zeta} \right) ^* & i\text{ odd}. \end{cases}
\end{align}
For even $i$, we choose $u_i(x)=u(x),v_i(x)=1$ and for odd $i$, we choose $v_i(x)= u(x) , u_i(x)=1$. We always have $s_i=1$. Now we can apply \autoref{iterated comp kernel} and get for any $x \not =y$ that

\begin{align}
&\left \lvert \operatorname{iker} \left(  \left( H_\varepsilon  M_{I,\zeta} \right) \left( H_\varepsilon  M_{I,\zeta} \right) ^* \right) ^{n_0}(x,y) \right \rvert \\
\le & F\left( \zeta\right)\frac{ b_0(x,y)+ \exp \left( -\lambda {B_0} \lVert x-y \rVert^2 \right) }{(1+ \lVert x \rVert)^{2n_0 \varepsilon}}.
\end{align}
The function $F$ is in $\Lp^\infty_{loc}(\mathbb C \setminus (2I+1){B_0})$. This integral kernel is in $\Lp^2$, as $2n_0 \varepsilon >1$. The $b_0$ term only appears for $n_0=1$, as for $n_0 >1$, we get $s <0$, which corresponds to a bounded $b_s$.
\end{proof}

We will now prove some useful methods to deal with the differentials we will have to apply in order to use \autoref{Sob emb Schatten2}. We will first see that, in a way, $M_{I,\zeta}$ can take two differentials, and then we will see how to move further differentials past $M_{I,\zeta}$ and $H_{\varepsilon,x_0}$.

Let $j_1,j_2 \in \{1,2\}$ and $h \in \{\pm 1\}$. Then we observe the commutator relation
\begin{align}
\left[ \left(- i \nabla - h A_0 \right)_{j_1} , \left( - i \nabla - A_0 \right)_{j_2} \right] &= i \frac{B_0} 2  \left( [ \nabla_{j_1}, (JX)_{j_2} ] +  h  [  (JX)_{j_1} , \nabla_{j_2} \right) \\
&= i \frac{B_0} 2 \left( J_{j_2 j_1} - h  J_{j_1 j_2} \right). \label{comm relations}
\end{align}
Here, $X$ refers to the multiplication operator associated to the identity on $\mathbb R^2$. 
As the matrix $J$ is skew-symmetric, this states that the so called covariant derivative $-i \nabla +A_0$ commutes with $-i \nabla-A_0$ and hence it commutes with the operators $H_0,P_l,M_{I, \zeta},T_l$ for any $ l \in \mathbb N$, cofinite subset $I \subset \mathbb N$ and any $\zeta \in \mathbb C \setminus {B_0} (2I+1)$.

For $h=+1$, however, it motivates the definition of the annihilation and construction operators. They are defined by
\begin{align}
a_{\pm} \coloneqq \frac 1 {\sqrt {B_0} } \left( \left(- i \nabla - A_0 \right)_1 \pm i   \left(- i \nabla - A_0 \right)_2 \right). 
\end{align}
Using \eqref{comm relations} for $j_1=1,j_2=2,h=1$ and \eqref{def J}, we observe
\begin{align}
{B_0} a_-a_+ =H_0 + B_0 , \quad{B_0} a_+ a_- = H_0 -B_0 , \quad a_+^* = a_-. \label{ a_+a_-}
\end{align}
This implies that $a_-$ is surjective and $a_+$ is injective. Let $l \in \mathbb N$. Then we have
\begin{align}
(H_0+B_0)a_-P_l = {B_0} a_-a_+a_- P_l = a_- (H_0- B_0) P_l= 2lB_0  a_- P_l. 
\end{align}
This states that $a_-P_l$ maps into the eigenspace of $H_0$ with eigenvalue $(2l-1)B_0$, which is the image of $P_{l-1}$. If $l>0$, as $a_-$ is surjective, it has to map the image of $P_l$ onto the image of $P_{l-1}$. With an analogous computation for $a_-^*=a_+$, we arrive at
\begin{align}
P_{l-1}a_-=P_{l-1}a_-P_l=a_-P_l. \label{ann + gen property}
\end{align}
We recall that the operator $M_{i,\zeta}$ has been defined in \eqref{def M_Iz}.
\begin{lemma} \label{2nd diff resolv}
For any $I \subset \mathbb N$ cofinite and $ j_1, j_2 \in\{1,2\}$, there is an $F \in L^\infty \left( \mathbb C \setminus B_0 (2I+1) \right)$, such that for any $\zeta \in \mathbb C \setminus B_0(2I+1)$, we have the estimates
\begin{align}
\left \lVert \left( - i \nabla - A_0 \right)_{j_1}  M_{I, \zeta}  \left( - i \nabla - A_0 \right)_{j_2} \right \rVert_\infty \le F( \zeta) , \\
\left \lVert \left( - i \nabla - A_0 \right)_{j_1}  \left( - i \nabla - A_0 \right)_{j_2}  M_{I, \zeta} \right \rVert_\infty \le F( \zeta).
\end{align}
\end{lemma}

\begin{proof} We will only prove the first claim, as the second follows completely analogous. 
As both components of $(- i \nabla -A_0)$  are linear combinations of $a_+,a_-$, it suffices to show that for any $h_1,h_2 \in \{ +,- \}$, we have the required estimate for the operator $a_{h_1} M_{I, \zeta} a_{h_2}$. Let $l \in \mathbb N$.  We consider the operator
\begin{align}
S_l \coloneqq  a_{h_1} M_{I, \zeta} a_{h_2} P_l.
\end{align}
For any $ k\in \mathbb N$, we define $k++=k+1$ and $k+-=k-1$. Using \eqref{ann + gen property}, we see
\begin{align}
a_{h_1} M_{I, \zeta} a_{h_2} P_l &= a_{h_1} M_{I, \zeta} P_{l+h_2}  a_{h_2} \\
&= \frac{ 1_I(l+h_2)}{ (2(l+h_2)+1)B_0 - \zeta}a_{h_1}P_{l+h_2} a_{h_1} \label{A 80}\\
&=\frac{ 1_I(l+h_2)}{ (2(l+h_2)+1)B_0 - \zeta} P_{l+h_1+h_2} a_{h_1}a_{h_2}.
\end{align}
We use the convention $\frac 0 0 =0$ in this proof. 
Hence, the family of operators $S_l$ satisfy the conditions of \autoref{orth op fam}. So we just need to bound the norm of $S_l$. Using \eqref{ a_+a_-}, we observe that for any $h \in \{+,- \} $ and $k \in \mathbb N$,
\begin{align}
\lVert a_hP_k \rVert^2 = \lVert P_k a_h^* a_h P_k \rVert = 2k+1+h.
\end{align}
Using \eqref{A 80}, this leads to
\begin{align}
\lVert S_l \rVert &= \frac{  1_I(l+h_2)}{ \lvert (2(l+h_2)+1)B_0 - \zeta \rvert } \lVert a_{h_1} P_{l+h_2} \rVert \lVert P_{l+h_2} a_{h_2} P_l \rVert \\
&= \frac{  1_I(l+h_2)}{ \lvert (2(l+h_2)+1)B_0 - \zeta \rvert } \sqrt { (2(l+h_2)+1+h_1) ( 2l+1+h_2) }  \le F( \zeta).
\end{align}
This finishes the proof.
\end{proof}

\begin{lemma} \label{commutator replacement}
Let $\gamma,n \in \mathbb N$ and assume that $(B_\varepsilon,V_\varepsilon)$ are $(\gamma,\varepsilon)$ tame. Let $\mathbb N \ni \gamma' \le \gamma$. Then there is a set of matrices $\left( N_ \mu \in \operatorname{Lin}\left(( \mathbb C^{2^{\gamma'}} , \mathbb C^{2^{\gamma'}}\right) \right)_{\mu \in \mathbb N^{n+1}, \lvert \mu \rvert =\gamma'}$, such that for any admissible $I, \zeta,x_0$, the identity
\begin{align}
\left( - i \nabla + A_0 \right) ^{\otimes \gamma'} \left( M_{I,\zeta} H_{\varepsilon, x_0}\right)^n&= \sum_{\mu \in \mathbb N ^{n+1}, \lvert \mu \rvert= \gamma'}N_\mu \left( \bigotimes_{j=1}^n  M_{I ,\zeta} H_{\varepsilon,x_0} ^{(\mu_j)}\right)\otimes \left(-i \nabla + A_0 \right) ^{\otimes \mu_{n+1}}
\end{align}
holds in the sense that both operators agree as continuous operators from the space $W^{\gamma', \infty}_{(\infty)}(\mathbb R^2 , \mathbb C)$ to the space $W^{0, \infty}_{(\infty)}\left (\mathbb R^2 , \mathbb C^{2^{\gamma'}}\right)$.
\end{lemma}

\begin{proof}
Let $d \in \mathbb N, h\in \{1,2\}$ with $0 \le d < \gamma$. We recall \eqref{def B_eps^j} to \eqref{def H_{eps,x0}}, and  \eqref{def H_{eps,x0} alt}. 
 We have $A_{\varepsilon,x_0}^{(d)} \in W^{\gamma-d,\infty}_{(\varepsilon)}(\mathbb R^2, \mathbb C^{2^{d+1}} )$ and $W_{\varepsilon,x_0}^{(d)}\in  W^{\gamma-d,\infty}_{(\varepsilon)}(\mathbb R^2, \mathbb C^{2^{d}} )$ by the assumptions and \autoref{dA def}.
Hence, by the product rule, we have for any $\mathbb N \ni \gamma' \le \gamma-d$ that the multiplication operators $A_{\varepsilon,x_0}^{(d)}$ and $W_{\varepsilon, x_0} ^{(d)}$ are continuous operators from $ W^{\gamma',\infty}_{(\infty)}(\mathbb R^2, \mathbb C )$ to the spaces $W^{\gamma',\infty}_{(\infty)}(\mathbb R^2, \mathbb C^{2^{d+1}})$ respectively  $ W^{\gamma',\infty}_{(\infty)}(\mathbb R^2, \mathbb C^{2^{d}})$. Furthermore, the operator  $(-i \nabla + A_0)$ obviously maps $ W^{\gamma'+1,\infty}_{(\infty)}(\mathbb R^2, \mathbb C )$ to   $W^{\gamma',\infty}_{(\infty)}(\mathbb R^2, \mathbb C^2)$ continuously for any $\gamma' \in \mathbb N$.   Finally, by \autoref{1.resolvent bound} and the fact that the covariant derivative $-i \nabla + A_0$ commutes with $M_{I, \zeta}$ by \eqref{comm relations}, for any $\gamma' \in \mathbb N$,  the operators $M_{I,\zeta}$ and $(-i\nabla -A_0 ) M_{I,\zeta}$ are continuous from $W^{\gamma',\infty}_{(\infty)}(\mathbb R^2, \mathbb C )$ to the spaces $W^{\gamma',\infty}_{(\infty)}(\mathbb R^2, \mathbb C)$, respectively $W^{\gamma',\infty}_{(\infty)}(\mathbb R^2, \mathbb C^2)$. These statements guarantee that every composition of operators we consider is well-defined in the claimed sense.

Now, by \eqref{comm relations} and \eqref{def H_{eps,x0} alt}, we have
\begin{align}
&\left(-i \nabla + A_0 \right)_h M_{I,\zeta} H^{(d)}_{\varepsilon,x_0}\\
 =&  M_{I,\zeta}\left(   \left(-i \nabla - A_0 \right) \cdot  \left(-i \nabla + A_0 \right)_h  A^{(d)}_{\varepsilon,x_0}+   \left(-i \nabla + A_0 \right)_h   W^{(d)}_{\varepsilon,x_0} \right) \\
= &M_{I,\zeta} \Big(  \left(-i \nabla - A_0 \right) \cdot A^{(d)}_{\varepsilon,x_0}  \left(-i \nabla + A_0 \right)_h   -i\left(-i \nabla - A_0 \right) \cdot \partial_ h A^{(d)}_{\varepsilon,x_0} \\
& \phantom{M_{I,\zeta} \Big(  } \quad  +W^{(d)}_{\varepsilon,x_0} \left(-i \nabla + A_0 \right)_h -i\partial_h  W^{(d)}_{\varepsilon,x_0} \Big) \\
= &M_{I, \zeta} H^{(d)}_{\varepsilon,x_0}   \left(-i \nabla + A_0 \right)_h + M_{I,\zeta} \left(e_h \cdot  H^{(d+1)}_{\varepsilon,x_0} \right).
\end{align}
The scalar product $e_h \cdot   H^{(d+1)}_{\varepsilon,x_0}$ reduces the first component of the tensor product $\left(\mathbb C^2 \right)^{\otimes(d+1)}$.

 Let $N'_d \colon \mathbb C^{2^d} \otimes \mathbb C^2 \to \mathbb C^2 \otimes \mathbb C^{2^d}$ that swaps the tensor factors ($u \otimes v \mapsto v \otimes u$). Then we have
\begin{align}
\left(-i \nabla + A_0 \right) \otimes  M_{I,\zeta} H^{(d)}_{\varepsilon,x_0} = N'_d M_{I, \zeta} H^{(d)}_{\varepsilon,x_0} \otimes \left(-i \nabla + A_0 \right) +M_{I, \zeta} H^{(d+1)}_{\varepsilon,x_0}.
\end{align}
The case $n=0$ or $\gamma'=0$ is tautological. The case $n=\gamma'=1$ follows, if we set $d=0$ above. Now we consider $n=1$ and the step $\gamma' \mapsto \gamma'+1 \le \gamma$,
\begin{align}
&\left( - i \nabla + A_0 \right) ^{\otimes (\gamma'+1)}M_{I,\zeta} H_{\varepsilon, x_0}\\
 =& \left( - i \nabla + A_0 \right) \otimes \sum_{\mu \in \mathbb N^2, \mu_1+\mu_2=\gamma'} N_\mu   M_{I ,\zeta} H_{\varepsilon,x_0} ^{(\mu_1)}\otimes  \left( - i \nabla + A_0 \right)^{\otimes \mu_2} \\
=& \sum_{\mu \in \mathbb N^2, \mu_1+\mu_2=\gamma'} \left( Id_{\mathbb C^2} \otimes N_ \mu \right) \left( - i \nabla + A_0 \right) \otimes M_{I ,\zeta} H_{\varepsilon,x_0} ^{(\mu_1)}\otimes  \left( - i \nabla + A_0 \right)^{\otimes \mu_2}  \\
=& \sum_{\mu \in \mathbb N^2, \mu_1+\mu_2=\gamma'}  \left( Id_{\mathbb C^2} \otimes N_ \mu \right) \left( N'_{\mu_1}     M_{I, \zeta} H_{\varepsilon, x_0} ^{(\mu_1)} \otimes \left( - i \nabla + A_0 \right) + M_{I, \zeta} H_{\varepsilon, x_0} ^{(\mu_1+1)}  \right)\otimes \left( - i \nabla + A_0 \right)^{\otimes \mu_2} \\
=& \sum_{\mu \in \mathbb N^2, \mu_1+\mu_2=\gamma'+1} N_\mu   M_{I ,\zeta} H_{\varepsilon,x_0} ^{(\mu_1)}\otimes  \left( - i \nabla + A_0 \right)^{\otimes \mu_2}. \\
\end{align}
In the last step, we used the inductive definition
\begin{align}
N_{(\mu_1,\mu_2)} \coloneqq \left( Id_{\mathbb C^2} \otimes N_ {(\mu_1, \mu_2-1)} \right)     \left( N'_{\mu_1} \otimes Id_{\mathbb C^{\mu_2}}\right) +\left( Id_{\mathbb C^2} \otimes N_ {(\mu_1-1, \mu_2)} \right) . 
\end{align}
To conclude the proof, we do an induction on $n$ over the statement of the lemma. The idea is to use the induction hypothesis and then the case $n=1$. We omit the details, as it works pretty similar to the induction on $\gamma'$. The only annoying part is creating a recursive description for the $N_\gamma$s. But we have no use for 
such a description.
\end{proof}

We can now prove \autoref{theorem final summand}.
\begin{proof}[Proof of \autoref{theorem final summand}]
We begin by conjugating with the unitary operator $U_{x_0}$, that we have defined in \autoref{U_{x_0}}. Hence, as the $p$-Schatten quasi norm is unitarily equivalent, we have
\begin{align}
&\left \lVert \int_\Gamma 1_{[0,1]^2-x_0} \left(M_{\mathbb N ,\zeta} H_\varepsilon \right)^n \frac 1 {H-\zeta}   \left(   H_\varepsilon M_{\mathbb N ,\zeta}  \right)^n 1_{L \Lambda^\complement} d \zeta \right \rVert_p \\
=& \left \lVert \int_\Gamma 1_{[0,1]^2}  \left(M_{\mathbb N ,\zeta} H_ {\varepsilon,x_0} \right)^n  U_{x_0}\frac 1 {H-\zeta}   \left(   H_\varepsilon M_{\mathbb N ,\zeta}  \right)^n 1_{L \Lambda^\complement}  U_{x_0}^{-1} d \zeta \right \rVert_p \label{final summand conjugated}
\end{align}
Let $q$ satisfy $\frac 1 q + \frac 1 2 = \frac 1 p$. As $p > \frac 2 { \gamma +3}$, we have $q> \frac 2 { \gamma+2}$. Hence, we can apply \autoref{Sob emb Schatten2} with $\gamma+2$ and the property Hölder I (see \autoref{pnorm facts}) to get the upper bound
\begin{align}
\eqref{final summand conjugated}\le &C \left \lVert  \int_\Gamma \left(M_{\mathbb N ,\zeta} H_ {\varepsilon,x_0} \right)^n  U_{x_0}\frac 1 {H-\zeta}   \left(   H_\varepsilon M_{\mathbb N ,\zeta}  \right)^n 1_{L \Lambda^\complement}  U_{x_0}^{-1} d \zeta \right \rVert_{S_2(\Lp^2(\mathbb R^2), H^{\gamma+2}([0,1]^2))} \\ 
\le& C \int_\Gamma \left \lVert  \left(M_{\mathbb N ,\zeta} H_ {\varepsilon,x_0} \right)^n  U_{x_0}\frac 1 {H-\zeta}   \left(   H_\varepsilon M_{\mathbb N ,\zeta}  \right)^n 1_{L \Lambda^\complement}  U_{x_0}^{-1}  \right \rVert_{S_2(\Lp^2(\mathbb R^2), H^{\gamma+2}([0,1]^2))} d \zeta. \label{last summand int est}
\end{align}
The last step relies on the fact that the Hilbert--Schmidt norm (2-Schatten norm) is a norm and not just a quasi-norm. Now it suffices to bound the integrand uniformly on the integration path. For this, we first use the quasi-isometry $D_{\gamma+2}$ constructed in \autoref{H quasi iso}. Hence, we have
\begin{align}
& \left \lVert  \left(M_{\mathbb N ,\zeta} H_ {\varepsilon,x_0} \right)^n  U_{x_0}\frac 1 {H-\zeta}   \left(   H_\varepsilon M_{\mathbb N ,\zeta}  \right)^n 1_{L \Lambda^\complement}  U_{x_0}^{-1}  \right \rVert_{S_2(\Lp^2(\mathbb R^2), H^{\gamma+2}([0,1]^2))} \\
\le & C \sum_{\gamma'=-2}^\gamma  \left \lVert \left( - i \nabla +A_0 \right)^{\otimes (\gamma'+2)} \left(M_{\mathbb N ,\zeta} H_ {\varepsilon,x_0} \right)^n  U_{x_0}\frac 1 {H-\zeta}   \left(   H_\varepsilon M_{\mathbb N ,\zeta}  \right)^n 1_{L \Lambda^\complement}  U_{x_0}^{-1}  \right \rVert_{S_2(\Lp^2(\mathbb R^2), \Lp^2([0,1]^2)}  \\
=& C\sum_{\gamma'=-2}^\gamma  \left \lVert1_{[0,1]^2} \left( - i \nabla +A_0 \right)^{\otimes (\gamma'+2)} \left(M_{\mathbb N ,\zeta} H_ {\varepsilon,x_0} \right)^n  U_{x_0}\frac 1 {H-\zeta}   \left(   H_\varepsilon M_{\mathbb N ,\zeta}  \right)^n 1_{L \Lambda^\complement}  U_{x_0}^{-1}  \right \rVert_2 \\ 
\le &C \sum_{\gamma'=-2}^\gamma  \left \lVert1_{[0,1]^2} \left( - i \nabla +A_0 \right)^{\otimes (\gamma'+2)} \left(M_{\mathbb N ,\zeta} H_ {\varepsilon,x_0} \right)^n \right \rVert_\infty  \left \lVert U_{x_0} \right \rVert_\infty  \left \lVert \frac 1 {H-\zeta}   \left(   H_\varepsilon M_{\mathbb N ,\zeta}  \right)^n 1_{L \Lambda^\complement}   \right \rVert_2  \left \lVert U_{x_0}^{-1}  \right \rVert_\infty \\
= &  C\sum_{\gamma'=-2}^\gamma \left \lVert1_{[0,1]^2} \left( - i \nabla +A_0 \right)^{\otimes (\gamma'+2)} \left(M_{\mathbb N ,\zeta} H_ {\varepsilon,x_0} \right)^n \right \rVert_\infty  \left \lVert \frac 1 {H-\zeta}\right \rVert_\infty  \left \lVert   \left(   H_\varepsilon M_{\mathbb N ,\zeta}  \right)^n 1_{L \Lambda^\complement}   \right \rVert_2 \label{last summand split eq}
\end{align}
The third step relies on applications of Hölder I (see \autoref{pnorm facts}). The last step uses that $U_{x_0}$ is unitary on $\Lp^2(\mathbb R^2)$ and another application of Hölder I. The conjugation with $U_{x_0}$ was only needed fo the first term. It does make a difference there, as $U_{x_0}$ is not unitary on $H^{\gamma+2}([0,1]^2)$ and does not commute with $D_{\gamma+2}$.

We begin with the last factor in \eqref{last summand split eq}. As we are still free to choose $n \in \mathbb N$, we can assume $n>2$. We use the kernel estimate in \autoref{H eps M kernel} and \autoref{iterated comp kernel}, similar to the proof of the second result of \autoref{H eps M kernel} to arrive at the following estimate for any $x,y \in \mathbb R^2$:
\begin{align}
\left \lvert \operatorname{iker} \left( H_\varepsilon M_{ \mathbb N, \zeta} \right)^n (x,y) \right \rvert \le F( \zeta) \frac{ \exp \left( - \lambda \lVert x-y \rVert^2 \right)}{(1+\lvert y \rVert)^{n \varepsilon}}.
\end{align}
Now we let $n \varepsilon > 1 +\nu$. Then using the Hilbert--Schmidt kernel identity, we have
\begin{align}
 \left \lVert   \left(   H_\varepsilon M_{\mathbb N ,\zeta}  \right)^n 1_{L \Lambda^\complement}   \right \rVert_2^2=& \int_{\mathbb R^2} dx \int_{L \Lambda^\complement} dy \left \lvert \operatorname{iker} \left( H_\varepsilon M_{ \mathbb N, \zeta} \right)^n (x,y) \right \rvert^2 \\
\le & F(\zeta)  \int_{L \Lambda^\complement} dy \frac 1 {(1+\lVert y \rVert )^{2+2\nu}} \\
\le & F( \zeta) L^{-2 \nu} \label{last summand last factor HS}
\end{align}
In the second step, we use that the Gauss kernel is integrable over $x$, that the integral is independent of $y$, and that $n \varepsilon > 1 +\nu$. The third step uses that there is some $r>0$ such that $D_r(0) \subset \Lambda$ and that $\nu>0$.

For the second factor in \eqref{last summand split eq}, we observe
\begin{align}
\left \lVert \frac 1 {H- \zeta} \right \rVert_\infty = \frac 1 {\operatorname{dist}(\zeta, \sigma(H))}, \label{last summand middle factor}
\end{align}
which is bounded along the path $\Gamma$.

For the first factor in \eqref{last summand split eq}, we first consider the case $\gamma' \ge0$. Here, 
we start by using \autoref{commutator replacement} with the parameters $\gamma'$ and $2$. Hence, we have
\begin{align}
&\left \lVert1_{[0,1]^2} \left( - i \nabla +A_0 \right)^{\otimes (\gamma'+2)} \left(M_{\mathbb N ,\zeta} H_ {\varepsilon,x_0} \right)^n \right \rVert_\infty  \\
=&\left \lVert1_{[0,1]^2} \left( - i \nabla +A_0 \right)^{\otimes 2} \otimes \! \sum_{\substack {\mu \in \mathbb N^{3}, \\ \lvert \mu \rvert=\gamma'}} \! N_\mu  M_{\mathbb N ,\zeta} H_{\varepsilon,x_0} ^{(\mu_1)}  \otimes  M_{\mathbb N ,\zeta} H_{\varepsilon,x_0} ^{(\mu_2)} \otimes \left(-i \nabla + A_0 \right) ^{\otimes \mu_3}   \left(M_{\mathbb N ,\zeta} H_ {\varepsilon,x_0} \right)^{n-2} \right \rVert_\infty \\
\le & C \sup_{\mu_1,\mu_2,\mu_3 \le \gamma} \left \lVert 1_{[0,1]^2} \left( -i \nabla + A_0 \right)^{\otimes 2} \otimes M_{\mathbb N, \zeta} \right \rVert_\infty \left \lVert  H_{\varepsilon,x_0} ^{(\mu_1)}  \otimes  M_{\mathbb N ,\zeta} H_{\varepsilon,x_0} ^{(\mu_2)} \right \rVert_\infty \\
& \quad \times \left \lVert \left(-i \nabla + A_0 \right) ^{\otimes \mu_3}   \left(M_{\mathbb N ,\zeta} H_ {\varepsilon,x_0} \right)^{n-2} \right \rVert_\infty. \label{long resolvent chain}
\end{align}
In the last step, we also used that $\mu_j \le \gamma' \le \gamma$. 
Now we need to estimate these three factors. We begin with the first one. 

By the proof of \autoref{H quasi iso}, we conclude that the map $D_2' \colon H^2([0,1]^2) \to \Lp^2([0,1]^2,\mathbb C^7)$ given by $u \mapsto \left( (-i \nabla - A_0)^{\otimes j} u \right)_{j=0}^2$ is a quasi-isometry. Hence, as the operators $M_{\mathbb N, \zeta}$ and $(-i \nabla -A_0) M_{\mathbb N, \zeta}$ are bounded by \autoref{1.resolvent bound} , and the operator $(-i \nabla -A_0)^{\otimes 2} M_{\mathbb N , \zeta}$ is bounded by \autoref{2nd diff resolv}, we have
\begin{align}
\left \lVert 1_{[0,1]^2} \left( -i \nabla +A_0 \right) ^{\otimes 2} M_{\mathbb N , \zeta} \right \rVert_\infty \le& C \left \lVert M_{\mathbb N, \zeta} \right \rVert_{S_\infty(\Lp^2(\mathbb R^2), H^2([0,1]^2)}\\
 \le& C \sum_{j=0}^2\left \lVert \left( -i \nabla -A_0 \right)^{\otimes j} M_{\mathbb N, \zeta} \right \rVert_\infty \le F(\zeta). \label{(-idel +a)^2M}
\end{align}
For any $\mathbb N \ni d \le \gamma$, the multiplication operators $A_{\varepsilon,x_0}^{(d)}, W_{\varepsilon,x_0}^{(d)}$ are bounded operators with a norm not depending on $x_0$. Furthermore, by \autoref{1.resolvent bound}, the operators $M_{\mathbb N, \zeta}, (-i \nabla -A_0) M_{\mathbb N, \zeta}$ and $M_{\mathbb N,\zeta} (-i \nabla -A_0)= \left( (-i \nabla- a_0) M_{I , \overline \zeta } \right)^*$ are bounded, and the operator $(-i \nabla -A_0) \otimes M_{\mathbb N, \zeta} (-i \nabla - A_0)$ is bounded by \autoref{2nd diff resolv}. Now, we use \eqref{def H_{eps,x0}} and \eqref{def H_{eps,x0} alt} to conclude
\begin{align}
&\left \lVert  H_{\varepsilon,x_0} ^{(\mu_1)}  \otimes  M_{\mathbb N ,\zeta} H_{\varepsilon,x_0} ^{(\mu_2)} \right \rVert_\infty \\
=& \left \lVert \left(   A_{\varepsilon,x_0} ^{(\mu_1)} \cdot (-i \nabla - A_0) +  W_{\varepsilon,x_0} ^{(\mu_1)} \right)  \otimes  M_{\mathbb N ,\zeta} \left(    (-i \nabla - A_0) \cdot A_{\varepsilon,x_0} ^{(\mu_2)} +  W_{\varepsilon,x_0} ^{(\mu_2)} \right)   \right \rVert_\infty \le  F(\zeta). \label{Heps M Heps}
\end{align}
We are left to estimate the expression in \eqref{long resolvent chain}. We rename $\mu_3$ to $d$ and do an induction over $d$ for $0 \le d \le \gamma$. Let $e  \in \mathbb N$ be minimal with $e \varepsilon \ge 1$. The claim of our induction is that for $n \ge d (e+2)+3$, we have the estimate
\begin{align}
 \left \lVert \left(-i \nabla + A_0 \right) ^{\otimes d}   \left(M_{\mathbb N ,\zeta} H_ {\varepsilon,x_0} \right)^{n-2} \right \rVert_\infty \le F(\zeta)(1+\lVert x_0 \rVert)^d, \label{ind hyp in last summand proof}
\end{align}
for some $F \in \Lp^\infty_{loc}(\mathbb C \setminus \sigma(H))$ depending on $n,d$. 
The induction start at $d=0$ only uses that  $ \lVert M_{I, \zeta}H_{\varepsilon,x_0}\rVert_\infty =\left  \lVert \left( H_{\varepsilon,x_0} M_{I, \overline{\zeta}}\right) ^*\right \rVert_\infty \le F(\zeta)$ by \autoref{H eps M kernel} and that the product of bounded operators is bounded. For the step $d \to d+1 \le \gamma$, we first use \autoref{commutator replacement} with the parameters $d$ and $e+2$. Hence, we have
\begin{align}
 &\left \lVert \left(-i \nabla + A_0 \right) ^{\otimes (d+1)}   \left(M_{\mathbb N ,\zeta} H_ {\varepsilon,x_0} \right)^{n-2} \right \rVert_\infty \\
=& \left \lVert \left(-i \nabla + A_0 \right) \otimes \sum_{\mu \in \mathbb N^{e+3}, \lvert \mu \rvert=d} N_\mu \bigotimes_{j=1}^{e+2} \left(M_{I ,\zeta} H_{\varepsilon,x_0} ^{(\mu_j)}\right)\otimes \left(-i \nabla + A_0 \right) ^{\otimes \mu_{e+3}} \left(M_{\mathbb N ,\zeta} H_ {\varepsilon,x_0} \right)^{n-4-e} \right \rVert_\infty \\
\le &C \sup_{\mu \in \mathbb N_{\le d}^{e+3} } \left \lVert \left(-i \nabla + A_0 \right) \otimes  \bigotimes_{j=1}^{e+2} \left(M_{I ,\zeta} H_{\varepsilon,x_0} ^{(\mu_j)}\right)\right \rVert_\infty \left \lVert   \left(-i \nabla + A_0 \right) ^{\otimes \mu_{e+3}} \left(M_{\mathbb N ,\zeta} H_ {\varepsilon,x_0} \right)^{n-4-e} \right \rVert_\infty \\
\le & C\sup_{\mu \in \mathbb N_{\le d}^{e+2} } \left \lVert \left(-i \nabla + A_0 \right)M_{\mathbb N, \zeta} \otimes   \bigotimes_{j=1}^{e} \left( H_{\varepsilon,x_0} ^{(\mu_j)}M_{I ,\zeta} \right)\right \rVert_\infty \left \lVert H_{\varepsilon,x_0}^{(\mu_{e+1})}\otimes  M_{\mathbb N, \zeta} H_{\varepsilon, x_0}^{(\mu_{e+2})} \right \rVert_\infty F(\zeta) (1+ \lVert x_0 \rVert)^d \\
\le & F(\zeta)(1+ \lVert x_0 \rVert)^d  \sup_{\mu \in \mathbb N_{\le d}^{e+2} } \left \lVert \left(-i \nabla + A_0 \right)M_{\mathbb N, \zeta} \otimes   \bigotimes_{j=1}^{e} \left( H_{\varepsilon,x_0} ^{(\mu_j)}M_{\mathbb N ,\zeta} \right)\right \rVert_\infty  . 
\end{align}
In the third step, we used the induction hypothesis and in the last step we used \eqref{Heps M Heps}. The remaining operator is just a product of integral operators.  The kernel of $(-i \nabla + A_0) M_{\mathbb N, \zeta}$ can be bounded using \autoref{1.resolvent bound}. Hence, we have
\begin{align}
&\left \lVert \operatorname{iker} (-i \nabla + A_0) M_{\mathbb N, \zeta} (x,y) \right \rVert \\
\le &\left  \lVert \operatorname{iker} (-i \nabla - A_0) M_{\mathbb N, \zeta} (x,y) \right \rVert +C \lVert x \rVert \lVert \operatorname{iker} M_{\mathbb N, \zeta} (x,y) \rVert\\
 \le& F(\zeta)(1+\lVert x \rVert ) \left(b_1(x,y)+\exp(- \lambda \lVert x-y \rVert^2)\right) . \label{6.51}
\end{align}
We used $b_0 \le C b_1$. 
We have estimated the integral kernels of the operators $H_{\varepsilon,x_0} ^{(\mu_j)}M_{\mathbb N ,\zeta}$ in \autoref{H eps M kernel}. Now, we can apply \autoref{iterated comp kernel}. As $\varepsilon<1$, we have $e>1$ and hence there is no singularity on the diagonal (the $b_s$ term is bounded). Hence, we have
\begin{align}
&\left \lVert \operatorname{iker} \left(-i \nabla + A_0 \right)M_{\mathbb N, \zeta} \otimes   \bigotimes_{j=1}^{e} \left( H_{\varepsilon,x_0} ^{(\mu_j)}M_{\mathbb N ,\zeta} \right) (x,y) \right \rVert \\
\le & F(\zeta) \frac{1+ \lVert x \rVert} {(1+ \lVert x+x_0 \rVert)^{e \varepsilon}} \exp(- \lambda \lVert x-y \rVert^2) \\
\le& F(\zeta) (1+ \lVert x_0 \rVert)  \exp(- \lambda \lVert x-y \rVert^2). \label{(-idel+A)M}
\end{align}
The final step relies on the fact $e \varepsilon \ge 1$. Using \autoref{nice>bounded}, we can conclude
\begin{align}
 \left \lVert \left(-i \nabla + A_0 \right)M_{\mathbb N, \zeta} \otimes   \bigotimes_{j=1}^{e} \left( H_{\varepsilon,x_0} ^{(\mu_j)}M_{\mathbb N ,\zeta} \right)\right \rVert_\infty \le F(\zeta) (1+ \lVert x_0 \rVert).
\end{align}
This finishes the induction over $d$. Hence, we have proven \eqref{ind hyp in last summand proof} and can continue the estimate in \eqref{long resolvent chain}. Using \eqref{(-idel +a)^2M} and \eqref{Heps M Heps}, we observe that for $0 \le \gamma' \le \gamma$, we have
\begin{align}
\left \lVert1_{[0,1]^2} \left( - i \nabla +A_0 \right)^{\otimes (\gamma'+2)} \left(M_{\mathbb N ,\zeta} H_ {\varepsilon,x_0} \right)^n \right \rVert_\infty \le F(\zeta) (1+ \lVert x_0 \rVert)^{\gamma} .
\end{align}
Now we need to consider the case $\gamma' \in \{-2,-1\}$. For  these, we estimate
\begin{align}
&\left \lVert1_{[0,1]^2} \left( - i \nabla +A_0 \right)^{\otimes (\gamma'+2)} \left(M_{\mathbb N ,\zeta} H_ {\varepsilon,x_0} \right)^n \right \rVert_\infty  \\
\le &  \left \lVert1_{[0,1]^2} \left( - i \nabla +A_0 \right)^{\otimes (\gamma'+2)}  M_{\mathbb N, \zeta} \right \rVert_\infty \left \lVert H_{\varepsilon, x_0} M_{\mathbb N, \zeta} H_{\varepsilon, x_0} \right \rVert_\infty \left \lVert M_{\mathbb N, \zeta} H_{\varepsilon,x_0} \right \rVert_\infty^{n-2}  \\
\le& F(\zeta) \le  F(\zeta) (1+ \lVert x_0 \rVert)^{\gamma} .
\end{align}
The operator $M_{\mathbb N, \zeta} H_{\varepsilon x_0}=H_{\varepsilon,x_0}M_{\mathbb N, \overline {\zeta}}^*$ has an operator norm $\le F(\zeta)$ by \autoref{H eps M kernel}, the middle factor is bounded by \eqref{Heps M Heps}, and the first factor is bounded by \autoref{1.resolvent bound} for $\gamma'=-2$ and by \eqref{6.51} for $\gamma'=-1$, in both cases the operator norm is $\le F(\zeta)$.

Now we have suitable upper bounds for the all factors in \eqref{last summand split eq}. The other factors are bounded by \eqref{last summand last factor HS} and \eqref{last summand middle factor}. Thus, we conclude
\begin{align}
& \left \lVert  \left(M_{\mathbb N ,\zeta} H_ {\varepsilon,x_0} \right)^n  U_{x_0}\frac 1 {H-\zeta}   \left(   H_\varepsilon M_{\mathbb N ,\zeta}  \right)^n 1_{L \Lambda^\complement}  U_{x_0}^{-1}  \right \rVert_{S_2(\Lp^2(\mathbb R^2), H^{\gamma+2}[0,1]^2)} \le F(\zeta) (1+ \lVert x_0 \rVert)^\gamma L^{- \nu}.
\end{align}
Using \eqref{last summand int est}, we have now finished this proof.
\end{proof}

We need one more technical lemma to prove \autoref{cube norm}. 
\begin{lemma} \label{hilfslemma}
Let $d \in \mathbb N, \kappa \in [0, \infty)$, and let $S$ be an integral operator on $\Lp^2(\mathbb R^2)$ satisfying for any $x,y \in \mathbb R^2$
\begin{align}
\lvert \operatorname{iker} S (x,y) \rvert \le C \frac{(1+\lVert x \rVert)^d}{(1+\lVert x +x_0 \rVert)^\kappa} \exp \left( -\lambda \lVert x- y \rVert^2 \right). 
\end{align}
Furthermore, let $\Omega \subset \mathbb R^2$ be bounded. Then there are constants $C,\lambda'$ such that for any $R \in [0,\infty)$, we have the estimate
\begin{align}
\left \lVert 1_{\Omega} S 1_{D_R^\complement(0)} \right \rVert_2 \le C \frac {\exp(- \lambda' R^2)} {(1+ \lVert x_0 \rVert)^\kappa}.
\end{align}
\end{lemma}
\begin{proof}
We use the Hilbert--Schmidt kernel property (see \autoref{pnorm facts}). Hence, by the unitary equivalence of the $p$-Schatten norms, we have
\begin{align}
\left \lVert 1_{\Omega} S 1_{D_R^\complement(0)} \right \rVert_2^2 =& \int_{\Omega} dx \int_{D_R^\complement (0) }dy\left \lVert  \operatorname{iker} S (x,y) \right \rVert^2 \\
\le & C\int_{\Omega} dx \int_{D_R^\complement (0) }dy \frac{(1+\lVert x \rVert)^{2d}}{(1+\lVert x +x_0 \rVert)^{2\kappa}} \exp \left( -2\lambda \lVert x- y \rVert^2 \right) \\
\le &C  \int_{\Omega} dx \int_{D_R^\complement (0) }dy  \frac{1}{(1+\lVert x_0 \rVert)^{2\kappa}} \exp(- \lambda \lVert y \rVert^2)  \\
\le &C   \frac{\exp (- 2\lambda' R^2) }{(1+\lVert x_0 \rVert)^{2\kappa}} 
\end{align}
The second step uses $x \in \Omega$ and \autoref{gauss disc}.  Then we used \autoref{gauss disc} again. This finishes the proof.
\end{proof}

\begin{proof}[Proof of \autoref{cube norm}]
We start off similarly to the proof of \autoref{theorem final summand}. In particular, we begin by conjugating with the unitary operator $U_{x_0}$, as defined in \autoref{U_{x_0}}. Hence, we have\footnote{We have already mentioned this equality in \eqref{conjugate U_x0}.}
\begin{align}
\left \lVert 1_{[0,1]^2-x_0} (T_l H_\varepsilon)^m P_l (H_\varepsilon T_l)^{k-m} 1_{D^\complement_R(x_0)}\right \rVert_p
=& \left \lVert 1_{[0,1]^2} (T_l H_{\varepsilon,x_0})^m P_l (H_{\varepsilon,x_0} T_l)^{k-m} 1_{D^\complement_R(0)}\right \rVert_p. \label{cube norm proof eq1}
\end{align}
Now, once again, let $q$ satisfy $\frac 1 q + \frac 1 2 = \frac 1 p$. As $p > \frac 2 { \gamma +3}$, we have $q> \frac 2 { \gamma+2}$. Hence, we can apply \autoref{Sob emb Schatten2} with $\gamma+2$ and the property Hölder I (see \autoref{pnorm facts}) to get the upper bound
\begin{align}
\eqref{cube norm proof eq1}\le &C \left \lVert   1_{[0,1]^2} (T_l H_{\varepsilon,x_0})^m P_l (H_{\varepsilon,x_0} T_l)^{k-m} 1_{D^\complement_R(0)} \right \rVert_{S_2(\Lp^2(\mathbb R^2), H^{\gamma+2}([0,1]^2))} \\ 
\le& C\sum_{\gamma'=-2}^\gamma\left \lVert   1_{[0,1]^2} (-i\nabla+A_0)^{\otimes(\gamma'+2)}(T_l H_{\varepsilon,x_0})^m P_l (H_{\varepsilon,x_0} T_l)^{k-m} 1_{D^\complement_R(0)} \right \rVert_{2}.
\end{align}
We used the quasi-isometry $D_{\gamma+2}$ as constructed in \autoref{H quasi iso}. 
We will now establish two kernel estimates that will be needed to finish this proof.

Looking at \eqref{pl kernel}, we observe that for any $d \in \mathbb N$ and $h \in \{0,1\}$, there are $\lambda,C \in \mathbb R^+$, depending on $B_0,l,d,h$, such that for any $x,y \in \mathbb R^2$, we have the upper bound
\begin{align} 
\left \lVert \left( -i \nabla -A_0\right)^{\otimes h}\otimes  \left(-i \nabla_x + A_0(x) \right)^{\otimes d} p_l(x,y) \right \rVert \le C (1+\lVert x \rVert)^{d+h} \exp\left( - \lambda \lVert x-y \rVert^2 \right). \label{pl is smooth}
\end{align}
Let $\mathbb N \ni j \le \gamma$. Then, using \eqref{def H_{eps,x0}}, we observe
\begin{align}
\left \lVert \operatorname{iker}H_{\varepsilon,x_0}^{(j)} \left(-i \nabla + A_0 \right)^{\otimes d} P_l(x,y) \right \rVert \le C\frac{ (1+\lVert x \rVert)^{d+1}}{(1+\lVert x +x_0  \rVert)^\varepsilon} \exp\left( - \lambda \lVert x-y \rVert^2 \right). \label{H eps smooth pl}
\end{align}
Now we consider the case $m=0$. Here, we can use \eqref{pl is smooth}, the kernel estimate for $H_{\varepsilon,x_0} T_l$, that is provided by \autoref{H eps M kernel}, and \autoref{iterated comp kernel} to arrive at
\begin{align}
\left \lVert \operatorname{iker}      (-i\nabla+A_0)^{\otimes(\gamma'+2)} P_l \left( H_{\varepsilon,x_0} T_l\right)^{k} (x,y) \right \rVert \le C \frac{(1+\lVert x \rVert)^{\gamma'+2} }{(1+\lVert x+ x_0 \rVert)^{k \varepsilon} } \exp\left( - \lambda \lVert x-y \rVert^2 \right).
\end{align}
As the kernel of $(-i\nabla +A_0)^{\otimes d}P_l$ has no singularity at the diagonal, the term $b_s$ can be ignored. By \autoref{hilfslemma}, we have now finished the case $m=0$.

Now we consider the case $m>0$ and $\gamma' \in \{-2,-1\}$. Here, we can use \autoref{1.resolvent bound} to get (compare \eqref{6.51})
\begin{align}
\left \lVert \operatorname{iker} \left(-i\nabla +A_0\right)^{\otimes (\gamma'+2)}T_l (x,y) \right \rVert \le C (1+ \lVert x \rVert )^{\gamma'+2} \left( b_1(x,y) + \exp(-\lambda \lVert x-y \rVert^2)\right).
\end{align}
With this kernel estimate, the one in \autoref{H eps M kernel}, and \eqref{H eps smooth pl} with $d=j=0$, we can employ \autoref{iterated comp kernel} to get
\begin{align}
\left \lVert \operatorname{iker} \left(-i\nabla +A_0\right)^{\otimes (\gamma'+2)} (T_l H_{\varepsilon,x_0})^m P_l (H_{\varepsilon,x_0} T_l)^{k-m} (x,y)\right \rVert \le C \frac{(1+\lVert x \rVert)^{\gamma'+3}}{(1+ \lVert x+x_0\rVert)^{k \varepsilon}} \exp \left( -\lambda \lVert x-y \rVert^2 \right).
\end{align}
Once again, as $p_l$ has no singularity at the diagonal, the term $b_s$ can be ignored and by \autoref{hilfslemma}, we have finished this case as well.

We are left with the case $m>0$ and $\gamma'  \ge 0$. Here, we first apply \autoref{commutator replacement} with the parameters $\gamma'$ and $m$. Hence, we have
\begin{align}
&\left \lVert 1_{[0,1]^2}  (-i\nabla+A_0)^{\otimes(\gamma'+2)}(T_l H_{\varepsilon,x_0})^m P_l (H_{\varepsilon,x_0} T_l)^{k-m} 1_{D^\complement_R(0)} \right \rVert_{2} \\
=& \left \lVert 1_{[0,1]^2}  (-i\nabla+A_0)^{\otimes2}\otimes \sum_{\mu \in \mathbb N^k, \lvert \mu \rvert=\gamma'} N_\mu \bigotimes_{j=1}^m(T_l H_{\varepsilon,x_0}^{(\mu_j)})\otimes  \left( - i \nabla+A_0 \right) ^{\otimes \mu_{m+1}} P_l (H_{\varepsilon,x_0} T_l)^{k-m} 1_{D^\complement_R(0)} \right \rVert_{2} \\
\le & C \sup_{\mu \in \mathbb N_ {\le \gamma}^{m+1} } \left \lVert 1_{[0,1]^2}  (-i\nabla+A_0)^{\otimes2}\otimes  \bigotimes_{j=1}^m(T_l H_{\varepsilon,x_0}^{(\mu_j)})\otimes  \left( - i \nabla+A_0 \right) ^{\otimes \mu_{m+1}} P_l (H_{\varepsilon,x_0} T_l)^{k-m} 1_{D^\complement_R(0)} \right \rVert_{2} \\
=&  C \! \sup_{\mu \in \mathbb N_ {\le \gamma}^{m+1} } \Bigg \lVert 1_{[0,1]^2}  (-i\nabla+A_0)^{\otimes2}T_l \otimes \\
& \phantom{ C \! \sup_{\mu \in \mathbb N_ {\le \gamma}^{m+1} } \Bigg \lVert } \bigotimes_{j=1}^{m-1}(H_{\varepsilon,x_0}^{(\mu_j)}T_l)\otimes H_{\varepsilon,x_0}^{(\mu_m)}  \left( - i \nabla+A_0 \right) ^{\otimes \mu_{m+1}} P_l (H_{\varepsilon,x_0} T_l)^{k-m} 1_{D^\complement_R(0)} \Bigg \rVert_{2} \label{m>0, g large} 
\end{align}
The operator $1_{[0,1]^2}(-i \nabla + A_0) ^{\otimes 2} T_l$ does not have a nice integral kernel. This is why we cannot directly get a kernel bound from this representation. Let $\varphi \in C^\infty_c(\mathbb R^2)$ be a smooth cutoff function satisfying $\varphi(x)=1$ for $x \in D_1(0)$, $\varphi(x)=0$ for $x \in D_2^\complement(0)$, and $0 \le \varphi(x) \le 1$ everywhere. We introduce the operators $T_{l,n}$ and $T_{l,f}$, which are defined by the integral kernels given for any $x,y \in \mathbb R^2$ with $x \neq y$  by
\begin{align}
t_{l,n}(x,y)\coloneqq &\varphi(x-y) t_l(x,y), \\
t_{l,f}(x,y)\coloneqq &(1-\varphi(x-y)) t_l(x,y).
\end{align}
Obviously, $T_{l,n}+T_{l,f}=T_l$. Furthermore, for any $d\in \{0,1,2\}$, the operator $(-i\nabla - A_0 )^{\otimes d} T_{l,f}$ has a nice integral kernel satisfying
\begin{align}
\left \lVert \operatorname{iker} (-i\nabla - A_0 )^{\otimes d} T_{l,f}(x,y) \right \rVert \le C \exp(- \lambda \lVert x-y \rVert^2)
\end{align}
by \autoref{1.resolvent bound}. 
This implies the kernel estimate
\begin{align}
\left \lVert \operatorname{iker} (-i\nabla + A_0 )^{\otimes 2} T_{l,f}(x,y) \right \rVert \le C (1+\lVert x \rVert)^2\exp(- \lambda \lVert x-y \rVert^2). \label{ T_lf kernel}
\end{align}
Hence, the operator $1_{[0,1]^2}(-i\nabla + A_0)^{\otimes 2} T_{l,f}$ is bounded. 
The operator $(-i\nabla-A_0)^{\otimes d} T_l$ is bounded by \autoref{1.resolvent bound} for $d=0,1$ and by \autoref{2nd diff resolv} for $d=2$. Hence, the operator $1_{[0,1]^2}(-i\nabla + A_0)^{\otimes 2} T_l$ is bounded. By the triangle inequality, we can conclude that the operator $1_{[0,1]^2}(-i\nabla + A_0)^{\otimes 2} T_{l,n}$ is bounded. Furthermore, we have the identity 
\begin{align}
1_{[0,1]^2}(-i\nabla + A_0)^{\otimes 2} T_{l,n} = 1_{[0,1]^2}(-i\nabla + A_0)^{\otimes 2}1_{[-1,2]^2}  T_{l,n} =1_{[0,1]^2}(-i\nabla + A_0)^{\otimes 2}  T_{l,n} 1_{[-3,4]^2}. \label{fin propagation}
\end{align}
The value at $x \in [0,1]^2$ of $(-i \nabla +A_0)f$ only depends on $f$ in an arbitrary small neighbourhood of $x$, which proves the first identity. The second identity follows by the construction of $T_{l,n}$ as an integral operator with a kernel that vanishes if $\lVert x-y \rVert\ge 2$.

We will now estimate the kernel of the operator in \eqref{m>0, g large}, where we replace the first $T_l$ by $T_{l,f}$. The kernels of the operators $H_{\varepsilon,x_0}^{(\mu_j)}T_l$ and $H_{\varepsilon,x_0}T_l$ can be bounded by \autoref{H eps M kernel},  the kernel of $H_{\varepsilon,x_0}^{(\mu_m)}  \left( - i \nabla+A_0 \right) ^{\otimes \mu_{m+1}} P_l $ has been bounded in \eqref{H eps smooth pl}, and the kernel of $(-i\nabla+A_0)^{\otimes2}T_{l,f}$ has been bounded in \eqref{ T_lf kernel}. Hence, we can apply \autoref{iterated comp kernel} to arrive at
\begin{align}
&\left \lVert  \operatorname{iker}(-i\nabla+A_0)^{\otimes2}T_{l,f} \otimes  \bigotimes_{j=1}^{m-1}(H_{\varepsilon,x_0}^{(\mu_j)}T_l)\otimes H_{\varepsilon,x_0}^{(\mu_m)}  \left( - i \nabla+A_0 \right) ^{\otimes \mu_{m+1}} P_l (H_{\varepsilon,x_0} T_l)^{k-m} (x,y) \right \rVert \\
\le & C \frac{(1+\lVert x \rVert)^{\gamma+3}}{(1+ \lVert x+x_0\rVert)^{k \varepsilon}} \exp \left( -\lambda \lVert x-y \rVert^2 \right). \label{Tlf full kernel}
\end{align}
Once more, the $b_s$ term can be ignored as the operator $H_{\varepsilon,x_0}^{(\mu_m)}  \left( - i \nabla+A_0 \right) ^{\otimes \mu_{m+1}} P_l $ has no singularity at the diagonal and by \autoref{hilfslemma}, this establishes the required estimate. 

We are only left with the term in \eqref{m>0, g large}, where we replace the first $T_l$ by $T_{l,n}$. Here, we can use \eqref{fin propagation} to see
\begin{align}
&\left \lVert 1_{[0,1]^2}  (-i\nabla+A_0)^{\otimes2}T_{l,n} \otimes  \bigotimes_{j=1}^{m-1}(H_{\varepsilon,x_0}^{(\mu_j)}T_l)\otimes H_{\varepsilon,x_0}^{(\mu_m)}  \left( - i \nabla+A_0 \right) ^{\otimes \mu_{m+1}} P_l (H_{\varepsilon,x_0} T_l)^{k-m} 1_{D^\complement_R(0)} \right \rVert_2 \\
= &   \Bigg \lVert 1_{[0,1]^2}  (-i\nabla+A_0)^{\otimes2}T_{l,n} 1_{[-3,4]^2}  \otimes \\
&\phantom{    \Bigg \lVert  }\quad \bigotimes_{j=1}^{m-1}(H_{\varepsilon,x_0}^{(\mu_j)}T_l)\otimes H_{\varepsilon,x_0}^{(\mu_m)}  \left( - i \nabla+A_0 \right) ^{\otimes \mu_{m+1}} P_l (H_{\varepsilon,x_0} T_l)^{k-m} 1_{D^\complement_R(0)} \Bigg \rVert_2 \\
\le &C  \left \lVert 1_{[0,1]^2}  (-i\nabla+A_0)^{\otimes2}T_{l,n} \right \rVert_\infty \\
& \quad \times  \left \lVert 1_{[-3,4]^2}  \otimes  \bigotimes_{j=1}^{m-1}(H_{\varepsilon,x_0}^{(\mu_j)}T_l)\otimes H_{\varepsilon,x_0}^{(\mu_m)}  \left( - i \nabla+A_0 \right) ^{\otimes \mu_{m+1}} P_l (H_{\varepsilon,x_0} T_l)^{k-m} 1_{D^\complement_R(0)} \right \rVert_2.
\end{align}
The operator $1_{[0,1]^2}  (-i\nabla+A_0)^{\otimes2}T_{l,n}$ is bounded. For the remaining part, we estimate the kernel. This is incredibly similar to \eqref{Tlf full kernel}. 
The kernels of the operators $H_{\varepsilon,x_0}^{(\mu_j)}T_l$ and $H_{\varepsilon,x_0}T_l$ can be bounded by \autoref{H eps M kernel} and  the kernel of $H_{\varepsilon,x_0}^{(\mu_m)}  \left( - i \nabla+A_0 \right) ^{\otimes \mu_{m+1}} P_l $ has been bounded in \eqref{H eps smooth pl}. Hence, we can apply \autoref{iterated comp kernel} to arrive at
\begin{align}
&\left \lVert  \operatorname{iker} \bigotimes_{j=1}^{m-1}(H_{\varepsilon,x_0}^{(\mu_j)}T_l)\otimes H_{\varepsilon,x_0}^{(\mu_m)}  \left( - i \nabla+A_0 \right) ^{\otimes \mu_{m+1}} P_l (H_{\varepsilon,x_0} T_l)^{k-m} (x,y) \right \rVert \\
\le & C \frac{(1+\lVert x \rVert)^{\gamma+1}}{(1+ \lVert x+x_0\rVert)^{k \varepsilon}} \exp \left( -\lambda \lVert x-y \rVert^2 \right) .
\end{align}
For one final time, the $b_s$ term can be ignored as the operator $H_{\varepsilon,x_0}^{(\mu_m)}  \left( - i \nabla+A_0 \right) ^{\otimes \mu_{m+1}} P_l $ has no singularity at the diagonal and by \autoref{hilfslemma}, this establishes the required estimate. 

This brings this proof to a close.
\end{proof}

\appendix
\section{}

\begin{lemma} \label{U_{x_0}}
Let $x_0 \in \mathbb R^2$. Then there is a unitary operator $U_{x_0}\colon \Lp^2(\mathbb R^2) \to \Lp^2(\mathbb R^2)$, such that the following identities hold for any $f \in \Lp^\infty (\mathbb R^2)$, any $I \subset \mathbb N$ cofinite and any $\zeta \in \mathbb C \setminus B_0(2I+1)$:
\begin{align}
U_{x_0}f(X) U_{x_0}^{-1} =& f(X+x_0), \\
U_{x_0}  \left (-i\nabla - A_0\right) U_{x_0}^{-1} =& \left (-i\nabla - A_0 \right),\\
U_{x_0}  H_\varepsilon  U_{x_0}^{-1} =& H_{\varepsilon,x_0} \, ,\\
U_{x_0}  M_{I,\zeta}  U_{x_0}^{-1} =& M_{I,\zeta}\,.
\end{align}
Here, $X$ refers to the multiplication operator with the identity on $\mathbb R^2$ and $f(X)$ is defined by functional calculus and hence the multiplication operator with the function $f$.
\end{lemma}
The operators $H_\varepsilon, H_{\varepsilon,x_0}$, and $M_{I,\zeta}$ have been defined in \eqref{def Heps}, \eqref{def H_{eps,x0}}, and \eqref{def M_Iz}.

\begin{proof}
For any $x_0 \in \mathbb R^2$, we define the three unitary operators $U_{x_01} ,U_{x_02}, U_{x_0}$ by 
\begin{alignat}{3}
\forall x \in \mathbb R^2  \colon \quad &\left(U_{x_01} \psi \right)(x) \,&\coloneqq&\,\psi(x+x_0), \\  
\forall x \in \mathbb R^2  \colon \quad &\left(U_{x_02} \psi \right)(x) \,&\coloneqq  &\,\psi(x) \exp \left(- i\frac {B_0} 2 \langle x \mid Jx_0 \rangle \right),
\end{alignat}
\begin{align}
U_{x_0}\coloneqq U_{x_01} U_{x_02} . 
\end{align}
As we can see, these operators and their inverses preserve $C^\infty_c(\mathbb R^2)$. Hence, it is sufficient to show that the claimed operator identites hold, when evaluated at a test function $\psi \in C^\infty_c( \mathbb R^2)$. 

We have
\begin{align}
U_{x_01}U_{x_02} f(X) U_{x_02}^{-1} U_{x_01}^{-1}  &= U_{x_01}  f(X) U_{x_01} ^{-1} \\
&= f(X+x_0). \label{conj1}
\end{align}

Now, we need to check how $(-i\nabla-\frac {B_0}2JX)$ behaves under conjugation with $U_{x_0}$. Hence, we get
\begin{align}
&\left( U_{x_02} \left (-i\nabla - \frac {B_0} 2 JX \right) U_{x_02}^{-1} U_{x_01}^{-1} \psi \right) (x) \\
=&  \exp \left(-i\frac {B_0} 2 \langle x \mid Jx_0 \rangle \right) \left(-i\nabla_x-  \frac {B_0} 2 Jx \right) \exp \left( i\frac {B_0} 2 \langle x \mid Jx_0 \rangle \right) \psi(x-x_0) \\
=& \left(-i\nabla_x- \frac {B_0} 2 Jx \right) \psi(x-x_0) + \psi(x-x_0) \left(-i\nabla_x \right) \left( i \frac {B_0}2 \langle x \mid Jx_0 \rangle \right)  \\
=& \left(-i\nabla_x- \frac {B_0} 2 J(x-x_0) \right) \psi(x-x_0) \\
=&\left(  U_{x_01}^{-1}  \left (-i\nabla - \frac B 2 JX \right) \psi \right) (x).
\end{align}
In the second step, we used the product and chain rule and the exponentials cancel. The interior derivative is then resolved in the next step.

In conclusion, we have
\begin{align}
U_{x_0}  \left (-i\nabla - A_0\right) U_{x_0}^{-1} = \left (-i\nabla - A_0 \right).
\end{align}
This implies 
\begin{align}
U_{x_0}  T_l U_{x_0}^{-1} = T_l.
\end{align}
Together with \eqref{conj1}, this implies the identity
\begin{align}
U_{x_0}  H_\varepsilon  U_{x_0}^{-1} =& H_{\varepsilon,x_0} .
\end{align}
This finishes the proof.
\end{proof}

\begin{lemma}\label{grid integral}
Let $\Omega \subset \mathbb R^n$ be measurable and let $f \colon \Omega \to \mathbb C$ be integrable. Then we have the identity
\begin{align}
\int_{[0,1)^n} \sum_{z \in \mathbb Z^n, z+h_0 \in \Omega} f(z+h_0) dh_0 = \int_\Omega f(x) dx.
\end{align}
\end{lemma}

\begin{proof} We observe
\begin{align}
\int_{[0,1)^n} \sum_{z \in \mathbb Z^n, z+h_0 \in \Omega} f(z+h_0) dh_0 =& \int_{[0,1)^n} \sum_{z \in \mathbb Z^n}1_\Omega(z+h_0) f(z+h_0) dh_0 \\
=&\int_{\mathbb R^n} 1_\Omega(x) f(x) dx \\
=& \int_\Omega f(x) dx.
\end{align}
In the second step we used Fubini with $[0,1)^n \times  \mathbb Z^n = \mathbb R^n$.

\end{proof}

\begin{lemma} \label{Fubini1}
Let $f \in C^1(\mathbb R)$ with $ f' \le 0$ and $\lim_{t \to \infty} f(t)=0$, $h \in C^0(\mathbb R^2 ,\mathbb R)$ and $\Lambda \subset \mathbb R^2$ be measurable. 

Then we have
\begin{align} 
\int_\Lambda f(h(x) ) dx = \int_{\mathbb R} -f'(t) \lvert \{ x \in \Lambda \mid h(x) \le t \} \rvert dt .
\end{align}
As both integrands are positive, we do not need to require the existence of the integral, both sides being $\infty$ is an option.

\end{lemma}

\begin{proof}
We use the fundamental theorem of calculus and Fubini. As everything is positive, we can apply both theorems. Thus,
\begin{align}
\int_\Lambda f(h(x) ) dx =& \int_\Lambda dx \int_{h(x)}^\infty( -f'(t)) dt\\
=&\int_{\mathbb R^2}dx \int_{\mathbb R}dt 1_{\Lambda}(x) 1_{(h(x), \infty)}(t) (-f'(t))\\
=&\int_{\mathbb R}dt \int_{\mathbb R^2}dx 1_{\Lambda}(x) 1_{(h(x), \infty)}(t) (-f'(t))\\
=& \int_{\mathbb R} -f'(t) \left \lvert \{ x \in \Lambda \mid h(x) \le t \} \right  \rvert dt. \qedhere
\end{align}

\end{proof}

\begin{lemma} \label{Lipschitz slice}
Let $\Lambda \subset \mathbb R^2$ be a bounded Lipschitz region. Then there is a constant $C>0$, such that for any $r>0$
\begin{align}
\left \lvert \{ x \in \Lambda \mid \operatorname{dist} (x, \Lambda^\complement) \le r \} \right \rvert & \le Cr \label{inner near boundary}, \\
\left \lvert \{ x \in \Lambda^\complement \mid \operatorname{dist}(x, \Lambda) \le r\} \right \rvert & \le C(r+r^2) \label{outer near boundary}.
\end{align}
\end{lemma}

In both cases, for small $r$ we have an approximately linear dependency. In the first case, it is bounded by $\left \lvert \Lambda \right \rvert<\infty$ and in the second case it is contained in a ball of radius $r+r_0$, which explains the $r^2$ term.

\begin{lemma} \label{gauss disc}
Let $R,\lambda > 0$ be real numbers and $x_0,x \in \mathbb R^2$ with $\lVert x-x_0 \rVert \le R$. Then we have
\begin{align}
\exp ( -\lambda  \lVert x \rVert^2 )  &\le e^{\lambda R^2} \exp(-\frac \lambda 2 \lVert x_0 \rVert ^2), \\
\int_{D_R^\complement (0) }\exp ( -\lambda  \lVert x' \rVert^2 ) dx'& = \frac \pi \lambda \exp(-\lambda R^2).
\end{align}
\end{lemma}
For $\lVert x_0 \rVert \le R$, the estimate is trivial. Otherwise, the proof follows by taking the $\ln$, dividing by $\lambda$ and then completing the square.

\begin{lemma} \label{Bochner kernel}
For every $t \in (0,1)$, let $K_t \colon \Lp^\infty(\mathbb R^2) \to \Lp ^\infty(\mathbb R^2)$ be an operator with a \emph{nice} integral kernel $k_t: \mathbb R^2 \times \mathbb R^2 \to \mathbb C$. Assume, that for every $x \in \mathbb R^2$, the function $[0,1]\times \mathbb R^2 \to \mathbb C\colon (t,y) \mapsto k_t(x,y)$ is integrable, its integral is bounded independently of $x$, and the same holds for $x$ and $y$ reversed. Then we have
\begin{align}
\operatorname{iker}\left(\int_0^1 K_t dt \right)(x,y) =\int_0^1  k_t(x,y) dt.
\end{align}
\end{lemma}
\begin{proof}
The integral $\int_0^1 K_t dt$ exists as a Bochner integral with respect to the operator norm from $\Lp^\infty(\mathbb R^2)$ to $\Lp^\infty(\mathbb R^2)$ by the integrability assumptions on the kernel. 
Let $f \in C_c^0(\mathbb R^2)$. Then, for every $x \in \mathbb R^2$, we have 
\begin{align}
\left(\left(\int_0^1 K_t dt\right) f\right)(x) =& \left(\int_0^1 K_t f dt \right) (x) \\
=& \int_0^1 \left( \int_{ \mathbb R^2} k_t(x,y) f(y) dy \right) dt \\
=& \int_{\mathbb R^2} \left( \int_0^1 k_t(x,y) dt \right) f(y) dy .
\end{align}
The first step holds, as the Bochner integral commutes with the (linear, bounded) evaluation operator. The second step is the definition of $k_t$ and the last step is Fubini, as $f$ is bounded and we assumed $k_\blank(x,\blank)$ to be integrable for any $x\in \mathbb R^2$. The same holds, if $x$ and $y$ are reversed, hence this is a \emph{nice} integral kernel again.
\end{proof}

\begin{lemma} \label{orth op fam}
For any $k \in \mathbb Z^+$, let $S_k$ be an operator on the Hilbert space $\mathcal H$ and assume that for any $k \neq l$, the conditions $S_k^*S_l=0$ and $S_kS_l^*=0$ hold. Then we have
\begin{align}
\left\lVert \sum_{k \in \mathbb Z^+}S_k \right  \rVert_\infty &= \sup_{k \in \mathbb Z^+} \lVert S_k \rVert_\infty . 
\end{align}
\end{lemma}
\begin{proof}
For $l \in \mathbb Z^+$, let $\mathcal H_l$ be the orthogonal complement of the kernel of $S_l$ and define $\mathcal H_0 \coloneqq \bigcap_{l \in \mathbb Z^+} \ker(S_l)$. 
The condition $S_kS_l^*=0$ tells us that the spaces $\mathcal H_l$ and $\mathcal H_k$ are orthogonal. Hence, we have $\mathcal H = \bigoplus_{l \in \mathbb N} \mathcal H_l$.  Let $\Psi \in \mathcal H$. Then we can consider the expansion along this direct sum and get a sequence $\left( \Psi_l \in \mathcal H_l \right)_{l \in \mathbb N}$. We consider
\begin{align}
\left \lVert \left( \sum_{k \in \mathbb Z^+}S_k \right) \Psi \right \rVert^2 &= \left \lVert \sum_{k \in \mathbb Z^+}S_k  \Psi_k \right \rVert^2 \\
&=\sum_{k \in \mathbb Z^+} \left \lVert S_k \Psi_k \right \rVert^2 \\
& \le \sum_{k \in \mathbb Z^+} \lVert S_k \rVert _\infty^2 \lVert \Psi_k \rVert^2 \\
&\le \sup_{k \in \mathbb Z^+} \lVert S_k \rVert_\infty^2 \sum_{k \in \mathbb Z^+} \lVert \Psi_k \rVert^2 \\
&=  \sup_{k \in \mathbb Z^+} \lVert S_k \rVert_\infty^2 \ \lVert \Psi \rVert^2.
\end{align}
The condition $S_k^*S_l=0$ implies that the images of $S_k$ and $S_l$ are orthogonal. We used this in the second step. For the other inequality, for any $l \in \mathbb Z^+$, we observe
\begin{align}
\lVert S_l \Psi \rVert^2 = \lVert S_l \Psi_l \rVert^2 \le \sum_{k \in \mathbb Z^+} \lVert S_k \Psi_k \rVert^2 =\left \lVert \left( \sum_{k \in \mathbb Z^+}S_k \right) \Psi \right \rVert^2 .
\end{align}
This finishes the proof.
\end{proof}

\begin{definition}
Let $\gamma \in \mathbb N$, and let $\Omega \subset \mathbb R^2$ be open with Lipschitz-boundary. Then we define the Hilbert space $H^{\gamma}(\Omega)$ as the closure of $C^\infty(\overline \Omega, \mathbb C)$ under the norm
\begin{align}
\lVert u \rVert_{H^{\gamma}(\Omega)}^2 \coloneqq \sum_{0 \le \gamma' \le \gamma} \left \lVert \nabla ^{\otimes \gamma'} u \right \rVert_{\Lp^2(\Omega)}^2. \label{eq H norm def}
\end{align}
We also write $H^\gamma(\overline \Omega )$ for $H^\gamma(\Omega)$. 
\end{definition}
The more commonly used norm
\begin{align}
 u \mapsto & \sqrt{ \sum_{\alpha \in \mathbb N^2, \lvert \alpha \rvert \le \gamma}  \lVert \partial ^\alpha u \rVert_{\Lp^2(\Omega) } ^2  } \label{H norm alt1}
\end{align}
is equivalent to \eqref{eq H norm def}.

\begin{lemma} \label{H quasi iso}
Let $\gamma \in \mathbb Z^+$. Then the map $D_\gamma \colon H^{\gamma}(\mathbb [0,1]^2) \to \Lp^2 \left ([0,1]^2 , \mathbb C^{2^{\gamma+1}-1}\right ) $ given by 
\begin{align}
u \mapsto  \left( (-i \nabla +A_0)^{\otimes \gamma'} u \right)_{\gamma'=0}^\gamma 
\end{align}
is a quasi-isometry, meaning that there is a constant $1<C<\infty$ such that for any $u \in H^{\gamma}([0,1]^2)$, we have
\begin{align}
\frac 1 C \lVert u \rVert_{H^{\gamma}([0,1]^2)}  \le  \lVert D_{\gamma} u \rVert_{\Lp^2\left ([0,1]^2, \mathbb C^{2^{\gamma+1}-1}\right )  }     \le C \lVert u \rVert_{H^{\gamma}([0,1]^2)} .
\end{align}
\end{lemma}
The multiplication operator $A_0$ has been defined in \eqref{def A0}.
\begin{proof}
Let $0 \le \gamma' \le \gamma$ be a natural number and let $\kappa \in \{1,2\}^{\gamma'}$ be a multiindex. Now we can multiply out and simplify:
\begin{align}
\left( (-i \nabla +A_0 )^{\otimes \gamma'}u(x) - (-i)^{\gamma'} \nabla ^{\otimes \gamma'} u(x) \right)_\kappa  = \sum_{k,l \in \mathbb N, k+l<\gamma'} r_{k,l,\kappa}(x) \partial_1^k \partial_2^l u(x),
\end{align}
where $r_{k,l,\kappa}$ is a polynomial of degree at most $\gamma'-k-l$ that does not depend on $u$. As it is a polynomial, it is bounded on $[0,1]^2$. This leads to the upper bound 
\begin{align}
\left \lVert (-i \nabla +A_0)^{\otimes \gamma'} u - (-i)^{\gamma'} \nabla^{\otimes \gamma'} u \right \rVert_{L^2\left ([0,1]^2, \mathbb C^{2^{\gamma'}}\right )} \le C \lVert u \rVert_{H^{\gamma'-1}([0,1]^2)} \label{A.45}
\end{align}
for any $0<\gamma' \le \gamma$. This specific estimate is needed for the lower bound. For the upper bound, we can just put the $\nabla^{\otimes \gamma'}u$ on the other side and get
\begin{align}
\left \lVert (-i \nabla +A_0)^{\otimes \gamma'} u\right \rVert_{L^2\left ([0,1]^2, \mathbb C^{2^{\gamma'}}\right )} \le C \lVert u \rVert_{H^{\gamma'}([0,1]^2)}.
\end{align}
The claimed upper bound now follows by the triangle inequality.

For the lower bound, we let $C_0\ge 1$ be a constant that is sufficiently large to be the constant $C$ in \eqref{A.45} for any $1 \le \gamma' \le \gamma$. If there is a $0< \gamma' \le \gamma$ such that
\begin{align}
 \lVert \nabla^{\otimes \gamma'} u \rVert_{L^2\left ([0,1]^2, \mathbb C^{2^{\gamma'}}\right )} \ge 2 C_0 \lVert u \rVert_{H^{\gamma'-1}([0,1]^2)}, \label{C0 cases}
\end{align}
we choose $\gamma'$ maximal with this property. Otherwise, we set $\gamma'=0$. Now we observe that for any $\gamma\ge r > \gamma'$, we have 
\begin{align}
\lVert u \rVert_{H^r([0,1]^2)} ^2= \lVert u \rVert_{ H^{r-1}([0,1]^2)}^2 +\left   \lVert \nabla ^{\otimes r} u \right \rVert_{\Lp^2 \left ([0,1]^2, \mathbb C^{2^r}\right)} ^2 \le (4C_0^2+1)  \lVert u \rVert_{ H^{r-1}([0,1]^2)}^2 .
\end{align}
In conclusion, we have the estimate
\begin{align}
\lVert u \rVert_{H^{\gamma}([0,1]^2) } ^2 \le (4C_0^2+1)^{\gamma-\gamma'} \lVert u \rVert_{H^{\gamma'}([0,1]^2)} ^2 \le 2 (4C_0^2+1)^{\gamma-\gamma'}  \lVert \nabla^{\otimes \gamma'} u\rVert_{L^2\left ([0,1]^2, \mathbb C^{2^{\gamma'}}\right )}^2.
\end{align}
The last estimate relies on \eqref{C0 cases} and $C_0 \ge 1$, if $\gamma'>0$ .If $\gamma'=0$, then without the factor $2$, equality holds in the second inequality. By the triangle inequality, \eqref{A.45},  and \eqref{C0 cases} or trivially, if $\gamma'=0$, we get
\begin{align}
\lVert (-i \nabla +A_0 )^{\otimes \gamma'} u \rVert_{L^2\left ([0,1]^2, \mathbb C^{2^{\gamma'}}\right )} \ge \frac 1 2  \lVert  \nabla ^{\otimes \gamma'} u \rVert_{L^2\left ([0,1]^2, \mathbb C^{2^{\gamma'}}\right )} .
\end{align}
This finishes the lower bound and thus, the proof.
\end{proof}

The following proposition is a special case of Theorem 1 in \cite{Gramsch1968} by Gramsch.

\begin{proposition} \label{Sob emb Schatten1}
Let $\gamma \in \mathbb Z^+$, $\Omega \subset \mathbb R^2$ open, bounded and with $C^\infty$-boundary, and let $\infty> q> \frac 2{\gamma}$. Then the embedding
\begin{align}
\iota' \colon H^{\gamma}_0(\Omega) \to \Lp^2(\Omega)
\end{align}
is in the $q$-Schatten class. Here, $H_0^\gamma(\Omega)$ is the closure of $C^\infty_c(\Omega)$ under the norm of $H^\gamma(\Omega)$.
\end{proposition}

For the reader's convenience, we provide a different proof of this statement. This proof requires no regularity of $\partial \Omega$. [It can also be expanded to fractional exponent Hilbert spaces $H^s_0(\Omega)$.]
\begin{proof}
Let $- \Delta$ be the Dirichlet Laplacian on $\Omega$. Then the operator
\begin{align}
U \colon  H^{\gamma}_0(\Omega) \to \Lp^2(\Omega) , & \quad u \mapsto (1-\Delta)^{\frac \gamma 2 } u  
\end{align}
is bounded and its inverse is bounded as well. This is because the pullback of the norm on $\Lp^2(\Omega)$ via $U$ is equivalent to the norm on $H^{\gamma}_0(\Omega)$. To be precise, we have for any $u \in H^\gamma_0(\Omega)$
\begin{align}
\lVert U u \rVert_{\Lp^2(\Omega)}^2= \sum_{k=0}^\gamma\binom {\gamma} { k }\left \lVert \nabla^{\otimes k} u \right \rVert_{\Lp^2(\Omega,\mathbb C^{2^k})}^2.
\end{align}
This can be verified on the dense subset $C^\infty_c(\Omega)$ by partially integrating.

Now we consider the operator $V \colon \Lp^2(\Omega) \to H^{\gamma}_0(\Omega)$, given by $u \mapsto U^{-1} u\in H^{\gamma}_0(\Omega) \subset \Lp^2(\Omega)$. We want to estimate the $q$-Schatten norm of $V$. We define
\begin{align}
N(\lambda)= \# \{ \lambda' \le \lambda \colon \lambda \text{ is an eigenvalue of } - \Delta \}.
\end{align}
By Weyl's law, we conclude that there is a constant $C$, depending on $\Omega$, such that
\begin{align}
N(\lambda) \le C (1+ \lambda) ,
\end{align}
for any $\lambda \in \mathbb R^+$. Now we can write
\begin{align}
\lVert V \rVert_q^q =& \int_{\mathbb R^+ } (1+ \lambda)^{-q \frac {\gamma}2 } dN(\lambda) \\
=& \lim_{R \to \infty} \int_0 ^R (1+ \lambda)^{-q\frac {\gamma}2 } dN(\lambda) \\
=& \lim_{R \to \infty} \left(  N(R) (1+ R)^{-q\frac {\gamma }2 }  +q \frac {\gamma}2 \int_0^R (1+\lambda)^{-q\frac {\gamma }2-1 } N( \lambda) d \lambda \right) \\
\le & C \lim_{R \to \infty} \left( (1+R)^{1-q\frac {\gamma}2 } + \int_0^R (1+ \lambda)^{-q \frac {\gamma }2} d \lambda \right) \\
\le & C \lim_{R \to \infty} \left(  1+(1+R)^{1-q\frac {\gamma}2 } \right) \le C .
\end{align}
The final estimate relies on the condition $q > \frac 2{\gamma}$. Now, we just use that $U$ and $U^{-1}$ are bounded operators to get
\begin{alignat}2
\lVert \iota' \rVert_q &= \lVert VU \rVert_q &\le \lVert V \rVert_q \lVert U \rVert_\infty .
\end{alignat}
This finishes the proof.
\end{proof}

We want to apply the statement for the space $H^\gamma ( [0,1]^2)$. Neither Gramsch's result nor our proof is sufficient for that application. Hence, we need a slight extension.
\begin{corollary} \label{Sob emb Schatten2}
Let $\gamma \in \mathbb Z^+$, $\Omega \subset \mathbb R^2$ open with Lipschitz-boundary, and $\infty> q> \frac 2{\gamma}$. Then the embedding
\begin{align}
\iota \colon H^{\gamma}(\Omega) \to \Lp^2(\Omega)
\end{align}
is in the $q$-Schatten class.
\end{corollary}

\begin{remark}
In Proposition 2.1 in \cite{BS77}, Birman and Solomyak have shown an estimate of the singular values depending on the differentiability of the kernel. From that, one can see that for any Hilbert--Schmidt operator $S \colon \Lp^2(\mathbb R^2) \to H^\gamma\left([0,1]^2 \right)$, the operator $S$ is in the $p$-Schatten class for any $p> \frac 2 {\gamma+1}$. This statement also follows from our corollary here. 

We decided not to use Birman and Solomyak's result directly, as it is convenient for us to have this statement in the operator setting. Furthermore, we can directly use the quasi-isometry $D_\gamma$, that we constructed in \autoref{commutator replacement}.
\end{remark}

\begin{proof}[Proof of \autoref{Sob emb Schatten2}]
Let $\Omega' \supset \overline{\Omega}$ be an open ball. 
As $\Omega$ has Lipschitz-boundary, there is a continuous extension operator,
\begin{align}
E \colon H^{\gamma}(\Omega) \to H_0^{\gamma}(\Omega') .
\end{align}
 One such operator can be constructed as a composition of a multiplication operator with a smooth cutoff function and the extension operator constructed by Stein in Theorem 5 in \cite{STEIN1970}. 
Furthermore, there obviously is the continuous restriction operator
\begin{align}
R \colon \Lp^2(\Omega') \to \Lp^2(\Omega).
\end{align}
Hence, the operator
\begin{align}
\iota = R \iota' E
\end{align}
is in the $q$-Schatten class by \autoref{Sob emb Schatten1}.
\end{proof}

\section{Proof of \autoref{dA def}} \label{Appendix B}

\begin{proof}[Proof of \autoref{dA def}]
We recall
\begin{align}
g(x) \coloneqq \int_{\mathbb R^2} \frac { Jy} { 2 \pi \lVert y \rVert^2 } f (x-y) dy.
\end{align}

The last property will be seen by bounding this integral. $C_\varepsilon$ will be a constant depending only on $\varepsilon$, that may change from line to line. To begin with we have the bound
\begin{align}
\lVert g(x) \rVert\le& \int_{\mathbb R^2} \frac 1 {\lVert y \rVert} \frac {C} {(1+ \lVert y-x \rVert )^{1+\varepsilon}} dy \\
\le & \int_{D_{2 \lVert x \rVert}(0)}  \frac 1 {\lVert y \rVert} \frac {C} {(1+ \lVert y-x \rVert )^{1+\varepsilon}} dy  \\
&+ \int_{D^\complement_{2 \lVert x \rVert}(0)} \frac 1 {\lVert y \rVert} \frac {C} {(1+\lVert y-x \rVert )^{1+\varepsilon}} dy \\
\le & \int_{D_2(0)} \frac 1 { \lVert y \rVert }  \frac{C} {(\frac 1 {\lVert x \rVert }+\lVert y- e_1  \rVert) ^{1+\varepsilon}} \lVert x \rVert^{2-1-1-\varepsilon} dy \\
&+ \int_{D^\complement_{2 \lVert x \rVert}(0)} \frac 1 {\lVert y \rVert} \frac {C} {(1+\lVert y/2 \rVert )^{1+\varepsilon}} dy \\
\le &C \min \left \{   \lVert x \rVert^{-\varepsilon} ,      \lVert x \rVert    \right \} \\
&+ C   \max\{ \lVert x \rVert,1 \} ^{-\varepsilon}+  C_\varepsilon   1_{D_1(0)}(x) \\
\le & \frac{ C} {(1+ \lVert x \rVert)^{ \varepsilon}}.
\end{align}
In the second to last step, we got the first minimum by ignoring either of the summands in the denominator of the bounded domain integral and for the second part we just did a different bound on the annulus from $\lVert x \rVert $ to $1$, if $\lVert x \rVert<1$. This directly shows that $g \in W^{0, \infty}_{(\varepsilon)}(\mathbb R^2, \mathbb R^2)$.  For $\gamma>0$,  we can first use the result for $\partial_j f$  for $j=\{1,2\}$ and then use dominated convergence to see that $\partial_j g = (\partial_j f) * \frac{J\blank}{2\pi \lVert \blank \rVert^2} $. Hence, by an induction on $\gamma' \le \gamma$, we see that $g \in W^{\gamma, \infty}_{(\varepsilon)}(\mathbb R^2, \mathbb R^2)$ .

For the first two properties, we use the Fourier transform,
\begin{align}
\mathcal F (h)(\xi) \coloneqq \frac 1 {2 \pi} \int_{\mathbb R^2} h(x) \exp(-ix \cdot \xi) dx  , \quad \xi \in \mathbb R^2
\end{align}
for any $n \in \mathbb N$ and $ h \in \Lp^1 \cap \Lp^2( \mathbb R^2, \mathbb C^n)$. 
It can be expanded to tempered distributions and has the following properties for any $\xi \in \mathbb R^2$, tempered distributions $h,h_1,h_2$: 
\begin{align}
\mathcal F (\blank h(\blank))(\xi) &= i \nabla \mathcal F (h)(\xi), \\
\mathcal F(\nabla h (\blank))(\xi) &= -i \xi \mathcal F(h)(\xi), \\
\mathcal F(1) (\xi) &= 2\pi \delta_0 (\xi) ,\\
\mathcal F(h_1* h_2) (\xi) &= 2\pi \mathcal F(h_1)(\xi) \mathcal F(h_2)(\xi).
\end{align}
Here $\delta_0$ refers to the $\delta$-distribution at $0$. Furthermore, the Fourier transform is linear and invertible. As $f$ and $g$ are bounded, they are both tempered distributions. Now we can apply the Fourier transform to our first two claimed equations and are left to show
\begin{align}
- 2\pi i  J\xi  \cdot \mathcal F \left( \frac {J\blank}{2 \pi \lVert \blank \rVert ^2 } \right) (\xi) \mathcal F (f ) (\xi ) &= F (f ) (\xi ), \\
-2\pi i \xi \cdot  \mathcal F \left( \frac {J\blank}{2 \pi \lVert \blank \rVert ^2 } \right) (\xi) \mathcal F (f ) (\xi)&= 0.
\end{align}
Basically, this equation does not depend on $f$. Now we have to compute the Fourier transform of $\frac {Jx}{2\pi \lVert x \rVert^2}$,
\begin{align}
\mathcal F \left( \frac {J\blank}{2 \pi \lVert \blank \rVert ^2 } \right) (\xi)&= \frac 1 {2\pi} \left( J i \nabla (-\Delta)^{-1} \mathcal F( 1) \right) (\xi) \\
&= iJ  \nabla (-\Delta)^{-1} \delta_0 (\xi)\\
&= iJ\nabla \frac 1 {2\pi} \ln(\lVert \xi \rVert) \\
&= iJ \frac{\xi}{ 2 \pi \lVert \xi \rVert^2} \, .
\end{align}
Hence, we have
\begin{alignat}{2}
- 2\pi i  J\xi  \cdot \mathcal F \left( \frac {J \blank}{2 \pi \lVert \blank \rVert ^2 } \right) (\xi)&= -2\pi i J \xi \cdot  iJ \frac{\xi}{ 2 \pi \lVert \xi \rVert^2}\, &=1, \\
-2 \pi i \xi \cdot  \mathcal F \left( \frac {J \blank }{2 \pi \lVert \blank \rVert ^2 } \right) (\xi)&= -2\pi i \xi \cdot  iJ \frac{\xi}{ 2 \pi \lVert \xi \rVert^2}\,&=0.
\end{alignat}
This finishes the proof.
\end{proof}

\bibliography{mybib2020}{}
\bibliographystyle{plainurl}

\end{document}